\documentclass[11pt]{article}

\usepackage[numbers,sort]{natbib}

\usepackage{fullpage}

\usepackage[skip=0.5\baselineskip]{caption}

\usepackage{amsmath,amsfonts,amssymb,amsthm} % mathematical expression, symbols, fonts, and the whole theorem environment engine
\usepackage{mathtools}
\usepackage{tabularx,booktabs,multirow}

\usepackage{tikz}
\usetikzlibrary{arrows,decorations.pathmorphing,decorations.pathreplacing,backgrounds,positioning,fit,matrix}
\usetikzlibrary{shapes,calc}

\usepackage{pgfplots}
\usepackage{siunitx}
\usepackage{pgfplots}
\usepackage{pgfplotstable}

\usepackage{algorithm}
\usepackage[noend]{algpseudocode}
\algnewcommand\algorithmicinput{\textbf{Input:}}
\algnewcommand\algorithmicoutput{\textbf{Output:}}
\algnewcommand\Input{\item[\algorithmicinput]}
\algnewcommand\Output{\item[\algorithmicoutput]}
\algnewcommand\Continue{\textbf{continue}}

\usepackage[pagebackref]{hyperref}
\usepackage[sort&compress,nameinlink,noabbrev,capitalize]{cleveref}

\theoremstyle{plain} % typical theorem-style (bold header, italic font)
\newtheorem{theorem}{Theorem}[section]
\newtheorem{lemma}[theorem]{Lemma}
\newtheorem{corollary}[theorem]{Corollary}
\newtheorem{observation}[theorem]{Observation}
\newtheorem{proposition}[theorem]{Proposition}

\theoremstyle{definition} % typical definition-style (bold header, normal font)
\newtheorem{definition}[theorem]{Definition}

\theoremstyle{remark} % typical theorem-style (italic header, normal font)

\theoremstyle{plain}

\newcommand{\problemdef}[3]{
  \begin{center}
    \begin{minipage}{0.95\textwidth}
      \noindent
      \textsc{#1}
      
      \vspace{2pt}
      \setlength{\tabcolsep}{1pt}
      \begin{tabularx}{.95\textwidth}{@{}lX@{}}
        \textbf{Input:} 		& #2 \\
        \textbf{Task:} 	& #3
      \end{tabularx}
    \end{minipage}
  \end{center}
  \smallskip
}
\newcommand{\problem}[1]{\textsc{#1}}

\newcommand{\NN}{\mathbb{N}} % natural numbers
\newcommand{\bigO}{\mathcal{O}} % Landau symbol

\newcommand{\dist}{\operatorname{dist}}
\newcommand\abs[1]{\left|#1\right|} % absolute value
\newcommand{\tmin}{t_{\textup{min}}}
\newcommand{\SP}{\#P}
\newcommand{\SPC}{\SP-hard}

\newcommand{\TG}{\mathcal{G}}
\newcommand{\E}{\mathcal{E}}
\newcommand{\va}[2]{({#1},{#2})} % vertex appearance
\newcommand{\tedge}[3]{#1\overset{#2}{\operatorname{---}}#3}
\newcommand{\trans}[3]{#1\overset{#2}{\rightarrow}{#3}}
\newcommand{\tranz}[2]{\overset{#1}{\rightarrow}{#2}}

\newcommand{\sh}{\operatorname{sh}}
\newcommand{\fm}{\operatorname{fm}}
\newcommand{\fa}{\operatorname{fa}}
\newcommand{\p}{\operatorname{pfm}}

\newcommand\shortest[1]{\sigma^{{(\sh)}}_{#1}}
\newcommand\tshortest[2]{\sigma^{{(#1-\sh)}}_{#2}}

\newcommand\sigmageneric[1]{\sigma^{(\star)}_{#1}}
\newcommand\sigmaprefix[1]{\sigma^{{(\p)}}_{#1}}

\newcommand\pdgeneric[2]{\delta^{(\star)}_{#1}(#2)}
\newcommand\pdforemost[2]{\delta^{{(\fm)}}_{#1}(#2)}

\newcommand\pdshortest[2]{\delta^{{(\sh)}}_{#1}(#2)}
\newcommand\pdprefix[2]{\delta^{{(\p)}}_{#1}(#2)}

\newcommand\dgeneric[2]{\delta^{(\star)}_{#1\bullet}(#2)}

\newcommand\dshortest[2]{\delta^{{(\sh)}}_{#1\bullet}(#2)}

\newcommand\dprefix[2]{\delta^{{(\p)}}_{#1\bullet}(#2)}

\newcommand\predecessorsshortest[2]{P_{#1}^{\sh}(#2)}
\newcommand\predecessorstshortest[3]{P_{#2}^{#1-\sh}(#3)}

\newcommand\predecessorsprefix[2]{P_{#1}^{(\p)}(#2)}

\newcommand{\bcfm}{C^{{(\fm)}}_B}
\newcommand{\bcfa}{C^{{(\fa)}}_B}
\newcommand{\bcsh}{C^{{(\sh)}}_B}

\newcommand{\bcpfm}{C^{{(\p)}}_B}
\newcommand{\bc}{C^{(\star)}_B}

\newcommand{\betweenness}{C_B}
\newcommand{\betweennesssz}{\betweenness^{(SZ)}}

\newcommand*{\staticg}{G}
\newcommand*{\staticv}{\tilde{V}}
\newcommand*{\statice}{\tilde{E}}
\newcommand*{\weightfn}{\omega}
\newcommand*{\statictuplehelper}{\staticv, \statice}
\newcommand*{\statictuple}{(\statictuplehelper)}
\newcommand*{\statictuplew}{(\statictuplehelper, \weightfn)}

\newcommand*{\natinterval}[1]{\left[#1\right]}

\usepackage{authblk}

\begin{document}

\title{Algorithmic Aspects of Temporal Betweenness\thanks{Supported by the DFG,
project MATE (NI 369/17).}} %:\\ Theory \& Experiments}
\author{Sebastian Bu\ss}
\author{Hendrik~Molter}
\author{Rolf~Niedermeier}
\author{Maciej~Rymar}

\affil{Algorithmics and Computational Complexity, Faculty~IV, TU Berlin, Germany\\
\{buss,h.molter,rolf.niedermeier\}@tu-berlin.de, m.rymar@campus.tu-berlin.de}

\date{ }

\maketitle

\begin{abstract}
The \emph{betweenness centrality} of a graph vertex measures how often this
vertex is visited on shortest paths between other vertices of the graph. In the
analysis of many real-world graphs or networks, the betweenness centrality of a
vertex is used as an indicator for its relative importance in the network.
In particular, it is among the most popular tools in social network analysis. 
% Brandes' algorithm computes the betweenness centrality of all vertices in a static graph with $n$ vertices and $m$ edges in $\bigO(n\cdot m)$ time.
In recent years, a growing number of real-world networks has been modeled as
\emph{temporal graphs} instead of conventional (static) graphs. %because the latter are incapable of reflecting the dynamics of a network that changes over time.
In a temporal graph, we have a fixed set of vertices and there is a finite
discrete set of time steps and every edge might be present only at some time
steps. While shortest paths are straightforward to define in static graphs,
temporal paths can be considered ``optimal'' with respect to many different
criteria, including length, arrival time, and overall travel time (shortest,
foremost, and fastest paths).
This leads to different concepts of \emph{temporal betweenness centrality}, 
%of a vertex then can be defined based on any concept of optimal paths; %While this allows closer modeling of dynamic processes, 
posing new challenges on the algorithmic side. 
%Indeed, while in previous work ``temporal betweenness'' studies mostly 
%rely on (the simpler) computation of 
%temporal walks instead of temporal paths, we 
We provide a systematic study 
of temporal betweenness variants based on various concepts of optimal temporal
paths.

Computing the betweenness centrality for vertices in a graph is closely related
to counting the number of optimal paths between vertex pairs. While in static
graphs computing the number of shortest paths is easily doable in polynomial
time, we show that counting foremost and fastest paths is computationally
intractable (\SP-hard) and hence the computation of the corresponding temporal betweenness values is intractable as well. For shortest paths and two selected special cases of foremost paths, we devise polynomial-time algorithms
for temporal betweenness computation.
Moreover, we also explore the distinction between strict (ascending time labels)
and non-strict (non-descending time labels) time labels in temporal paths.
In our experiments with established real-world temporal networks, we demonstrate 
the practical effectiveness of our algorithms, compare the various betweenness concepts, and derive recommendations on their practical use.

\bigskip

\end{abstract}

\section{Introduction}\label{chap:intro}
Graph metrics such as betweenness centrality are studied and applied
in many application areas, including social and technological network analysis
\cite{leydesdorff_betweenness_2007,tang_temporal_2009}, wireless routing
\cite{daly_social_2007}, machine learning \cite{simsek_skill_2009}, and
neuroscience \cite{van_den_heuvel_aberrant_2010}.
The \emph{betweenness centrality} of a vertex in a graph measures how often this
vertex is visited by a shortest (or optimal) path. 
% many
% shortest (or optimal) paths in the graph go through this vertex. 
High
betweenness centrality scores are usually associated with vertices that can be seen as more important
for the network. In static graphs, betweenness centrality is a well-studied
concept. It is well-known that Brandes' algorithm~\cite{brandes_faster_2001}\footnote{According to Google Scholar, accessed
December 2020, the paper is cited more than 4200~times.} computes the
betweenness centrality of all vertices of a given (unweighted) static graph with
$n$~vertices and $m$~edges in~$\bigO(n\cdot m)$~time and~$\bigO(n+m)$~space.

In temporal graphs, that is, graphs with fixed vertex set and edge set(s)
varying over discrete time steps, the notion of betweenness centrality can be
defined in a similar fashion. However, there are more options how to choose ``optimal''
paths.
Depending on the application, a path may be optimal if it minimizes the number of edges (``shortest''), the arrival time (``foremost''),
or the overall travel time (``fastest''). For any of these path types we
can define and study a variant of temporal betweenness centrality. In addition,
combinations of optimality criteria such as \emph{shortest foremost}
temporal paths can be considered. Furthermore, we will distinguish between
temporal paths with strictly or non-strictly ascending time labels on the edges.
We investigate algorithmic aspects of temporal betweenness variants based on,
strict and non-strict, shortest, foremost, and fastest paths. 
In addition, we
also consider two subtypes of foremost paths, namely shortest foremost and
prefix-foremost paths. %and define temporal betweenness types based on those paths as well.

\paragraph{Related work.}
There is an enormous amount of work on the concept of betweenness 
centrality in static graphs, as already indicated by the huge 
citation numbers concerning Brandes path-breaking algorithm~\cite{brandes_faster_2001}. Betweenness centrality was defined in~1977 
by Freeman \cite{freeman_set_1977}. 
We refrain from further discussing the static case which is already treated 
in many textbooks.

The theory of temporal graphs is comparatively
young~\cite{holme2015modern,HS13,HS19,michail2016introduction,latapy2018stream}
but strongly growing in many directions.
We focus our discussion of related work on temporal walks, paths, and
the computation of 
temporal betweenness centrality. %of vertices in temporal graphs.

Bui-Xuan et al.\ \cite{xuan_computing_2003} did an early work on algorithms that
find optimal temporal paths (called ``journeys'' there). In particular, they
presented algorithms for shortest, fastest, and foremost temporal paths.
Afterwards, Wu et al.\ \citep{wu_efficient_2016} provided state-of-the-art
algorithms for optimal temporal paths. Based on breadth-first search which finds
shortest paths in static graphs, Wu et al.\ \citep{wu_efficient_2016} showed
that shortest, foremost, fastest and reverse-foremost (strict) temporal paths can be found in a similar
fashion.
%This is an important contribution because, unlike shortest static paths, not
%every optimal temporal path is composed of optimal subpaths.
Bentert et al.\ \cite{himmel_efficient_2019} and 
Casteigts et al.~\cite{HMZ19} expanded on the work of Wu et al.\
\citep{wu_efficient_2016} and studied a more complex variation of temporal paths and walks with
constraints on the waiting time in each vertex.
Bentert et al.\ \cite{himmel_efficient_2019} contributed efficient algorithms to
find optimal temporal \emph{walks} but Casteigts et al.~\cite{HMZ19} showed that
finding optimal temporal \emph{paths} is NP-hard in settings with upper bounds
on the waiting time.

While betweenness centrality in static graphs is a well studied concept, the study of betweenness centrality in temporal graphs is rather young. 
%As to betweenness centrality in temporal graphs, 
Tang et al.~\cite{tang_analysing_2010} argued that temporal graphs are more
suitable to represent the dynamics of social and technical networks and
introduced temporal variants of centrality metrics such as closeness and
betweenness centrality based on foremost temporal paths. Building on this, Tang
 et al.~\cite{tang_exploiting_2011} used their notion of temporal closeness to analyze the containment of malware in mobile phone networks.
Nicosia et al.\ \cite{nicosia_graph_2013} also discussed temporal variants of
betweenness and closeness centralities, as well as other temporal graph
metrics. They mostly give an overview on different definitions. 
Kim and Anderson \cite{kim_temporal_2012}
defined the temporal betweenness centrality of a vertex based on shortest paths
in the so-called \emph{static expansion}, which is a (static) directed graph that models the
connectivity properties of the corresponding temporal graph.
They give a polynomial-time algorithm for computing the
temporal betweeness values. 
Afrasiabi Rad et al.~\cite{rad2017computation}
studied foremost walk temporal betweenness and observed \SP-hardness, 
presented an exponential-time algorithm and conducted corresponding experiments.
Tsalouchidou et al.~\cite{tempbtw_2020} consider an
arbitrary linear combination of a path's length and duration as an optimality
criterion and compute the temporal betweenness centrality with respect to such
paths and the help of static expansions. 

We finish with pointing to several temporal graph
surveys~\cite{latapy2018stream,holme2015modern,HS13,HS19} that already provide
some definitions of temporal betweenness centrality.
We point out that most
works
on
temporal
betweenness~\cite{tempbtw_2020,alsayed2015betweenness,habiba2007betweenness,kim_temporal_2012,rad2017computation,tang_analysing_2010}
employ static expansions and do not directly work on the temporal graphs. Also, they usually do not make the point of distinguishing between strict and non-strict paths.
To the best of our knowledge, 
there is no work on systematically
classifying the computational complexity of temporal 
betweenness computation.

\paragraph{Our contributions.}
Our main research question is as follows: How hard---theoretically and
practically---is the computation of the different variants of temporal
betweenness centrality and what does this imply for their usefulness and
applicability in practice?
For the polynomial-time-solvable variants of temporal betweenness we aim to
provide algorithms inspired by Brandes' algorithm~\cite{brandes_faster_2001}
that directly work on the temporal graph rather than on its static expansion. 
Adapting algorithms
from the static directly to the temporal setting has proven successful also in
other contexts, such as for example clique
enumeration~\cite{Him+17,Ben+19,MNR20}.
We
empirically compare our direct approach to static-expansion-based algorithms and
observe that our algorithms are faster on large instances.

The various betweenness variants show remarkable differences in their
computational complexity: while some of them can be computed in polynomial
time, others are computationally hard. Note that computing
betweenness values is closely related to counting optimal paths since,
informally speaking, the betweenness value of a vertex quantifies how many
optimal paths go through this vertex.
We show that counting foremost and fastest temporal paths is \SPC, implying the
same hardness result for the computation of corresponding betweenness concepts.
Since temporal betweenness based on foremost temporal paths is arguably the best
motivated variant in many
application areas~\cite{tang_analysing_2010,rad2017computation} and we obtain
intractability, we investigate two modifications, namely shortest foremost and
(strict) prefix foremost temporal paths. For temporal betweeness based on these modified versions of foremost temporal paths as well as shortest temporal paths, we
obtain tractability results.

\begin{table}[t]
\centering
\begin{tabular}{lcc}
\toprule
& strict & non-strict\\
\midrule
Shortest & $\bigO(n^3\cdot T^2)$ &  $\bigO(n^3\cdot T^2)$ \\
Foremost & \SPC & \SPC \\
Fastest & \SPC  & \SPC \\
Prefix-foremost &  $\bigO(n \cdot M \cdot \log M)$ & \SPC \\
Shortest foremost &  $\bigO(n^3\cdot T^2)$ &  $\bigO(n^3\cdot T^2)$ \\
\bottomrule
\end{tabular}
\caption{An overview of the computational complexity of the temporal betweenness
variants we consider, both for the strict and the non-strict case. Here, $n$
refers to the number of vertices, $M$ refers to the total number of time edges, and $T$ refers to the total number of time
steps.}
\label{table:complexity}
\end{table}

In \cref{table:complexity}, we give an overview of our theoretical
findings.
We present formal (worst-case) computational hardness proofs
(\cref{ch:hardness}).
In the cases where we state polynomial-time computability, we provide algorithms
that compute the temporal betweenness scores of all vertices in a given
temporal graph (\cref{ch:adaptation,ch:algorithms}). When developing the
algorithms, our main approach was to direcly apply the ideas behind Brandes' algorithm~\cite{brandes_faster_2001} in the
temporal setting.

We remark that in contrast to the static setting, temporal \emph{walks}
(visiting vertices multiple times) can also be optimal for certain canonical
optimality criteria.
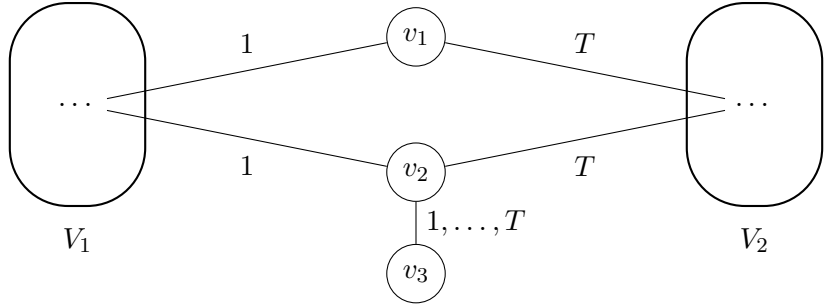
\begin{figure}[t]
	\center
	\begin{tikzpicture}[scale=.9]
		\path
			(0,.5) node[] {$V_1$}
			(10,.5) node[] {$V_2$}
			(5,3.5) node[circle, draw](1) {$v_1$}
			(5,1.5) node[circle, draw](2) {$v_2$}
			(5,0) node[circle, draw](3) {$v_3$}
			(0,2.5) node[](l) {\dots}
			(10,2.5) node[](r) {\dots};
		\draw (l) -- (1) node [midway, above = 3pt] {$1$};
		\draw (l) -- (2) node [midway, below = 3pt] {$1$};
		\draw (1) -- (r) node [midway, above = 3pt] {$T$};
		\draw (2) -- (r) node [midway, below = 3pt] {$T$};
		\draw (2) -- (3) node [midway,right] {$1,\dots,T$};
		\draw[thick,rounded corners=22pt] (-1,1) -- (-1,4) -- (1,4) -- (1,1) -- cycle;
		\draw[thick,rounded corners=22pt] (9,1) -- (9,4) -- (11,4) -- (11,1) -- cycle;
	\end{tikzpicture}
	\caption{Assume $V_1$ and $V_2$ are sets of vertices that are adjacent to~$v_1$ and~$v_2$ at time steps~$1$ and~$T$, respectively. Then~$v_1$ and~$v_2$ are vertices with high temporal betweenness based on temporal paths whereas $v_3$ has very low temporal betweenness based on temporal paths. However for e.g.\ foremost temporal walks, $v_3$ would also have a high temporal betweenness.}
	\label{figure:artifacts}
\end{figure}
In our definition of temporal betweenness centrality we use the number of
optimal temporal paths as opposed to the number of optimal temporal walks. This
is natural for static graphs because static shortest walks are always paths. In
the temporal case, this is not always true---it is possible that there is a
non-path temporal walk (visiting vertices multiple times) that arrives at the
same time as the foremost temporal path, so the number of foremost temporal
paths and foremost temporal walks between two vertices can be different.
Consider the graph shown in Figure \ref{figure:artifacts}. Clearly, every walk from the left half to the right passes either~$v_1$ or~$v_2$ and since the edges on the right are only present at exactly one time step, every walk going from left to right is a foremost walk. Intuitively, $v_1$ and~$v_2$ should have very similar, high betweenness scores, whereas $v_3$ should be close to zero. But if we use walks instead of paths for our definitions, we get a very high number of walks alternating between $v_2$ and $v_3$ before arriving on the right side, so $v_2$ and $v_3$ would get a high centrality score.
We conclude that paths are more suitable than walks for defining temporal betweenness centrality.
Furthermore, we distinguish between
strict and non-strict temporal paths.

We provide a thorough formal study of the computational complexity landscape of 
temporal betweenness centrality, altogether obtaining a fairly 
complete picture concerning the computation of temporal betweenness. 
We implemented and compared several of our algorithms (in terms of running
time, distribution of the betweenness values, and vertex rankings induced by
the betweenness values). 
We also compared our algorithms (in terms of running time) to the alternative
approach of computing betweenness values on a suitable static expansion.

From our experimental results we derive conclusions
for working with temporal betweenness centrality in applications.
Our freely available 
implementations and the experimental investigations 
are based on our theoretical findings and provide guidelines for 
future work with one of the perhaps most fundamental network analysis 
concepts in a temporal context.

An extended abstract of this paper appeared in the proceedings of the 26th ACM
SIGKDD International Conference on Knowledge Discovery \& Data Mining (KDD
'20)~\cite{kdd_paper}. This version contains full proof details and an empirical
comparison of our algorithms (in terms of running time) to the alternative
approach of computing betweenness values on a suitable static expansion.

\paragraph{Organization of the paper.}
In \cref{ch:prelims}, we provide the basic notation used in this
paper as well as some preliminary observations. In \cref{ch:hardness}, we prove
computational hardness for some variants of temporal betweenness (see \cref{table:complexity}) and in
\cref{ch:adaptation,ch:algorithms}, we provide theoretical foundations and
algorithms for the remaining temporal betweenness concepts. In \cref{ch:exp}, we
present our experimental evaluation of the algorithms described in
\cref{ch:algorithms}. We conclude in \cref{ch:conclusion}.

\section{Preliminaries \& Basic Observations}
	\label{ch:prelims}
In this section, we introduce the most important mathematical definitions and
terminology used in our work. We further present some basic observations.

\subsection{Temporal Graphs and Paths}

The fundamental mathematical object we are concerned with are \emph{temporal graphs}.

\begin{definition}[Temporal Graph]
An undirected \emph{temporal graph} is a triple~$(V,\E,T)$ such that 
$V$ is a set of vertices,
$\E \subseteq \{(\{u,v\},t) \mid u,v \in V,u\neq v,t \in [T]\}$ is a set of time edges, and
$T\in \mathbb{N}$, where $[T] = \{1,\dots,T\}$ is a set of time steps.
\end{definition}

For a temporal graph~$\TG$, we use~$V(\TG)$ to denote the set of vertices,
$E(\TG)$ for the set of time edges, and $E_t(\TG)$ to denote the set of edges
of~$\TG$ which are present at time step $t$, i.e., $E_t(\TG) := \{\{u,v\} \mid
(\{u,v\},t)\in E(\TG)\}$. For a time edge~$e$, we use $t(e)$ to denote the time
label of $e$. We call $V(\TG) \times [T]$ the set of \emph{vertex appearances}.

We only consider undirected temporal graphs. However, temporal paths and walks
are implicitly directed because of the ascending time labels. Hence, we need a
notion for directed \emph{transitions} on a temporal path or walk which indicate
not only which time edge is used but also in which direction.
%\begin{definition}[Transition]
For any time edge $e = (\{v,w\},t)$ we call $(v,w,t)$ the \emph{transition} from $v$ to $w$ at time step $t$. We call $v$ the starting point and $w$ the endpoint of the transition. 
%\end{definition}
Using this, we can now define temporal walks and temporal paths.

\begin{definition}[Temporal Walk]
A \emph{temporal walk} $W$ on a temporal graph $\TG$ from vertex $s$ to vertex $z$ is an ordered sequence of transitions $(e_1,\dots,e_k)\in E^k$ such that the endpoint of $e_i$ is the starting point of $e_{i+1}$ and $t(e_i) \leq t(e_{i+1})$ for each $i\in\{1,\dots,k-1\}$. We call a temporal walk \emph{strict} if $t(e_i) < t(e_{i+1})$ for each $i\in\{1,\dots,k-1\}$.
\end{definition}

A temporal walk may visit the same vertex more than once. In contrast to that,
a temporal \emph{path} visits each vertex at most once. This is analogous to the definitions for static graphs.

\begin{definition}[Temporal Path]
A \emph{temporal path} $P = (e_i)_{i\in[k]}$ is a temporal walk such that every
vertex $v\in V(\TG)$ is starting point of at most one transition $e_i$ and
endpoint of at most one transition~$e_{i'}$ for some $i,i'\in [k]$.
\end{definition}

For readability, we use the notation $\trans{v}{t}{w}$ instead of the triple $(v,w,t)$. Since the endpoint of a transition is equal to the starting point of the next one for any walk, we use a shortened notation omitting the doubled vertices. For instance, we denote the ``middle'' temporal path from $s$ to $z$ in \cref{figure:example1} by:
\[
	P = (\trans{s}{1}{b_1}\tranz{2}{b_2}\tranz{3}{b_3}\tranz{4}{z}).
\]

In this example, $P$ is a temporal \emph{path} since all involved vertices are
visited only once. Moreover,~$P$ is a \emph{strict} temporal path because the
time labels are strictly ascending.

\subsection{Optimality of Temporal Walks and Paths}
In static graphs, shortest paths are a central concept. In temporal graphs, there are different concepts of optimal paths. \cref{figure:example1} illustrates three of the most common optimization criteria: \emph{shortest}, \emph{foremost}, and \emph{fastest}~\cite{xuan_computing_2003}.
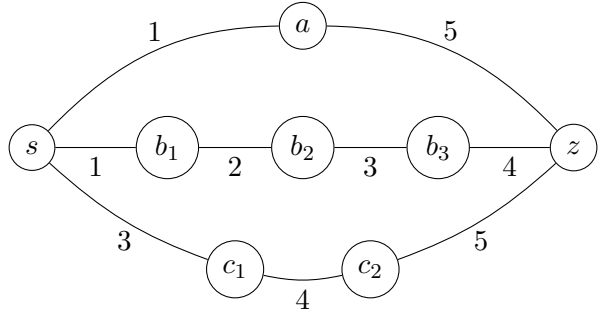
\begin{figure}[t]
	\centering
	\begin{tikzpicture}[scale=.9,yscale=1.2]
		\path
			(0,2) node[circle, draw](s) {$s$}
			(4,3.5) node[circle, draw](a) {$a$}
			(2,2) node[circle, draw](b1) {$b_1$}
			(4,2) node[circle, draw](b2) {$b_2$}
			(6,2) node[circle, draw](b3) {$b_3$}
			(3,0.5) node[circle, draw](c1) {$c_1$}
			(5,0.5) node[circle, draw](c2) {$c_2$}
			(8,2) node[circle, draw](z) {$z$};
		\draw (s) edge [bend left=20] node [midway, above] {1} (a);
		\draw (s) edge node [midway, below] {1} (b1);
		\draw (s) edge [bend right=10] node [midway, below] {3} (c1);
		\draw (b1) edge node [midway, below] {2} (b2);
		\draw (b2) edge node [midway, below] {3} (b3);
		\draw (c1) edge [bend right=10] node [midway, below] {4} (c2);
		\draw (b3) edge node [midway, below] {4} (z);
		\draw (c2) edge [bend right=10] node [midway, below] {5} (z);
		\draw (a) edge [bend left=20] node [midway, above] {5} (z);
	\end{tikzpicture}
	\caption{This temporal graph features a \emph{shortest}, a \emph{foremost}, and a \emph{fastest} temporal path from $s$ to $z$.} % Depending on the application, any of them may be considered optimal.}
	\label{figure:example1}
\end{figure}
\begin{itemize}
	\item $P_1 = (\trans{s}{1}{a}\tranz{5}{z})$ is shortest,
	\item $P_2 = (\trans{s}{1}{b_1}\tranz{2}{b_2}\tranz{3}{b_3}\tranz{4}{z})$ is foremost, and
	\item $P_3 = (\trans{s}{3}{c_1}\tranz{4}{c_2}\tranz{5}{z})$ is fastest.
\end{itemize}
%Each of the paths exemplifies one of the three most commonly used optimization criteria: the \emph{shortest} path is $P_1$ because it uses the least amount of edges. The \emph{foremost} path is $P_2$ because it has the earliest arrival time. Finally, $P_3$ is the \emph{fastest} path because it minimizes the overall transition time.
More formally, we use the following definitions:
\begin{definition}
Let $\TG = (V,\E,T)$ be a temporal graph. Let $s,z \in V$ and let $W$ be a temporal walk from $s$ to $z$. 
\begin{itemize}
\item $W$ is a \emph{shortest} walk if there is no walk~$W'$ from $s$ to $z$ such that $W'$ contains less transitions than $W$.
\item $W$ is a \emph{foremost} walk if there is no walk~$W'$ from $s$ to $z$ such that $W'$ has an earlier arrival time than $W$.
\item $W$ is a \emph{fastest} walk if there is no walk $W'$ from $s$ to $z$ such that the difference between arrival and start time is smaller for $W'$ than it is for $W$.
\end{itemize}
\end{definition}
A temporal \emph{path} is shortest, foremost, or fastest if it is a shortest,
foremost, or fastest temporal walk, respectively. Note that for the optimality
criteria ``foremost'' and ``fastest'' optimal temporal walks are not necessarily
temporal paths. Imagine a temporal graph where all time edges incident to $s$
have the same time label and also all time edges incident to~$z$ have the same
time label. Then every walk from $s$ to $z$ (if it exists) is foremost and
fastest, since all walks leave $s$ at the same time and arrive at $z$ at the same time.

Together with the distinction between strict and non-strict, this gives us
\emph{six} different types of optimal temporal paths. In the following, we use
the term ``$\star$-optimal'' temporal path, where $\star$ denotes the type.

%Analogous to Brandes \citep{brandes_faster_2001}, we use $\sigma_{sz}$ to denote the number of shortest paths from~$s$ to~$z$ in a static graph. 
Next, we introduce terminology and notation for counting optimal temporal paths. 
\begin{definition}
Let $\TG$ be a temporal graph. For any $s,z \in V(\TG)$, $\sigmageneric{sz}$ is the number of $\star$-optimal temporal paths from~$s$ to~$z$.
%\begin{itemize}
%	\item $\shortest{sz}$ is the number of \emph{shortest} paths from $s$ to $z$,
%	\item $\foremost{sz}$ is the number of \emph{foremost} paths from $s$ to $z$, and
%	\item $\fastest{sz}$ is the number of \emph{fastest} paths from $s$ to $z$.
%\end{itemize}
\end{definition}
We set $\sigmageneric{vv}:=1$.
%In addition to that, Brandes \citep{brandes_faster_2001} uses $\sigma_{sz}(v)$ to denote the number of shortest paths from~$s$ to~$z$ that pass through $v$ in a static graph. 
%Again, we adapt the notation to temporal graphs and the temporal optimality concepts.
We further introduce terminology and notation for counting optimal temporal paths that visit a certain vertex or a certain vertex appearance.
\begin{definition}
Let $v \in V$ be any vertex and let $t \in [T]$ be a time step. Then,
\begin{itemize}
	\item $\sigmageneric{sz}(v)$ is the number of $\star$-optimal paths that pass through~$v$, and
	\item $\sigmageneric{sz}(v,t)$ is the number of $\star$-optimal paths that
	arrive at~$v$ exactly at time step~$t$, that is, the paths that contain the
	transition $\trans{u}{t}{v}$ for some $u\in V$.
\end{itemize}
\end{definition}
We set $\sigmageneric{sz}(s):=\sigmageneric{sz}$ and
$\sigmageneric{sz}(z):=\sigmageneric{sz}$.
We define $\sigmageneric{sz}(s,0):=\sigmageneric{sz}$, $\sigmageneric{sz}(s,t):=0$ for all $t\neq 0$, and $\sigmageneric{sz}(v,0):=0$ for all $v\neq s$. 
Defining these corner cases as above allows us to keep certain proofs simpler by
avoiding to discuss the corner cases explicitely. Intuitively, we assume
that we have a dummy vertex appearance $(s,0)$ for every temporal path starting at
$s$ that we consider to be the vertex appearance that the path arrives at at
time step zero.

\subsection{Temporal Betweenness Centrality}
In static graphs, the \emph{betweenness centrality} of a vertex measures how
often this vertex is passed on shortest paths between pairs of vertices in the
graph. Freeman \cite{freeman_set_1977} defines the \emph{betweenness
centrality} $C_B(v)$ of a vertex~$v$ (in a connected graph) as \[ C_B(v) :=
\sum_{s \neq v \neq z} \frac{\sigma_{sz}(v)}{\sigma_{sz}}.\]

As we have seen above, there are different notions of optimal paths (for
example, fastest, shortest, foremost) in temporal graphs. Thus, there are several
options how to define \emph{temporal betweenness centrality} based on any of
these notions. 
Moreover, we do not want to assume that there is a temporal path from any vertex
to any other vertex in the graph. That is, we assume that there are vertex
pairs $s,z$ with $\sigma_{sz}=0$ which we want to leave out when summing over
all vertex pairs.
To formalize this, we use a \emph{connectivity matrix} $A$ of the temporal
graph: let $A$ be a
$|V|\times |V|$ matrix, where for every $v,w\in V$ we have that $A_{v,w}=1$ if
there is a temporal path from $v$ to $w$, and $A_{v,w}=0$ otherwise.
Note that $A_{s,z}=1$ implies that $\sigma_{sz}\neq 0$. 

Formally, temporal betweenness based on these different concepts of path
optimality is defined as follows.
\begin{definition}[Temporal Betweenness]
The \emph{temporal betweenness} of any vertex~$v\in V$ is given by:
\[
	C^{(\star)}_B(v) := \sum_{s \neq v \neq z \text{ and } A_{s,z}=1}
	\frac{\sigmageneric{sz}(v)}{\sigmageneric{sz}}. %,\\
	%\star \in \{\text{sh},\text{fa},\text{fm}\}.
\]
\end{definition}

For our work, we use a slightly different version of temporal betweenness
which is defined for vertex appearances and show that the temporal betweenness
as defined above can be easily computed from the modified version.
This allows us to simplify some of our proofs. 
We drop the condition that, when summing over all vertex pairs, these vertices
have to be different from the vertex of which we want to know the betweenness.
Formally, we define
\[
	\hat{C}^{(\star)}_B(v,t) := \sum_{s,z \in V \text{ and } A_{s,z}=1}
	\frac{\sigmageneric{sz}(v,t)}{\sigmageneric{sz}}. %,\\
	%\star \in \{\text{sh},\text{fa},\text{fm}\}.
\]
We can observe that using the connectivity matrix $A$ for the temporal
graph, we can compute the temporal betweenness~$C^{(\star)}_B(v)$ from the modified
temporal betweenness values of the vertex appearances~$(v,t)$. 
\begin{lemma}\label{lem:betversions}For any vertex~$v \in V$ it holds
\[
C^{(\star)}_B(v) = \sum_{t\in [T]\cup\{0\}} \hat{C}^{(\star)}_B(v,t) - \sum_{w\in V}
(A_{v,w} + A_{w,v}) +1.
\]
\end{lemma}
\begin{proof}
We show the claim as follows.
\begin{align*}
C^{(\star)}_B(v) &= \sum_{s \neq v \neq z \text{ and } A_{s,z}=1}
	\frac{\sigmageneric{sz}(v)}{\sigmageneric{sz}}\\
	& = \sum_{s,z \in V \text{ and }
	A_{s,z}=1} \frac{\sigmageneric{sz}(v)}{\sigmageneric{sz}} - \sum_{s \in V \text{ and }
	A_{s,v}=1} \frac{\sigmageneric{sv}(v)}{\sigmageneric{sv}} - \sum_{z \in V \text{ and }
	A_{v,z}=1} \frac{\sigmageneric{vz}(v)}{\sigmageneric{vz}} +
	\frac{\sigmageneric{vv}(v)}{\sigmageneric{vv}}\\
	& = \sum_{s,z \in V \text{ and }
	A_{s,z}=1} \sum_{t\in [T]\cup\{0\}} \frac{\sigmageneric{sz}(v,t)}{\sigmageneric{sz}} -
	\sum_{s \in V} A_{s,v} - \sum_{z \in V} A_{v,z} + 1\\
	&= \sum_{t\in [T]\cup\{0\}} \hat{C}^{(\star)}_B(v,t) - \sum_{w\in V}
(A_{v,w} + A_{w,v}) +1
\end{align*}
\end{proof}

\section{Computationally Hard Temporal Betweenness Variants}
\label{ch:hardness}
In this section, we present counting problems closely related to temporal
betweenness, including the computation of temporal betweenness itself, and show
that they are \SPC. In particular, we will show that counting all temporal
paths is \SPC\ for both strict and non-strict paths. The same holds true for
foremost and fastest temporal paths, but not for shortest temporal paths. Our
hardness results are based on the following two counting problems for which
Valiant \citep{valiant_complexity_1979} showed
\SP-completeness:\footnote{Intuitively speaking, \SP\ is the counting analogue
of NP for decision problems.} 

\problemdef{Paths}{A static graph $G=(V,E)$, two vertices
$s,z\in V$.}{Count the number of different paths from $s$ to $z$ in $G$.}
\problemdef{Imperfect Matchings}{A bipartite static graph.}{Count the
number of different matchings (of any size) in $G$.}

We use polynomial-time counting reductions\footnote{A \emph{polynomial-time
counting reduction} transforms instances $I_A$ of a counting problem $A$ to
instances $I_B$ of a counting problem $B$ such that number of solutions of $I_A$
can be computed in polynomial time from the number of solutions of $I_B$.} from the two problems
above to prove that counting problems related to temporal betweennes are \SPC.
More specifically, we show the \SP-hardness of the following problems:

\problemdef{(Strict) Temporal Paths}{A temporal graph $\TG=(V,\E,T)$, two
vertices $s,z\in V$.}{Count the number of (strict) temporal paths from $s$ to
$z$ in $\TG$.}
%\problemdef{Strict Temporal Paths}{A temporal graph $\TG$, vertices $s,z$.}{Count the number of strict temporal paths from $s$ to $z$ in $\TG$.}
\problemdef{Foremost (Strict) Paths}{A temporal graph $\TG=(V,\E,T)$, two
vertices $s,z\in V$.}{Count the number of foremost (strict) temporal paths from
$s$ to $z$ in $\TG$.} 
\problemdef{Fastest (Strict) Paths}{A temporal graph $\TG=(V,\E,T)$, two
vertices $s,z\in V$.}{Count the number of fastest (strict) temporal paths from $s$ to $z$ in $\TG$.}
\problemdef{(Foremost/Fastest) (Strict) Temporal Betweenness}{A temporal graph
$\TG=(V,\E,T)$, a vertex $v\in V$.}{Compute the betweenness centrality of $v$
in $\TG$.}

As also observed by Afrasiabi Rad et al.~\cite{rad2017computation}, for
non-strict paths the \SPC\ problem \problem{Paths} is contained as a special
case in \problem{Temporal Paths}, since any static graph can be transformed
into an equivalent temporal graph with lifetime $T=1$. Furthermore, each
temporal path in this instance is also foremost and fastest. Hence, we have the
following:

\begin{proposition}\label{prop:temppaths}
\problem{Temporal Paths}, \problem{Foremost Paths}, and \problem{Fastest Paths} are \SPC.
\end{proposition}

The counting problem \problem{Strict Temporal Paths} is \SPC\ as well, but the
proof is more demanding since strict paths are fundamentally different from
static paths, whereas non-strict paths could be regarded as a generalization of
static paths, allowing for the simple argument used above. We show the
\SP-hardness of \problem{Strict Temporal Paths} by a polynomial-time counting 
reduction from \problem{Imperfect Matchings}.
%Valiant \cite{valiant_complexity_1979} showed that \problem{Imperfect
% Matchings} is \SP-complete. %, but clearly, the two problems are equally hard, since every (bipartite) graph has exactly one empty matching.
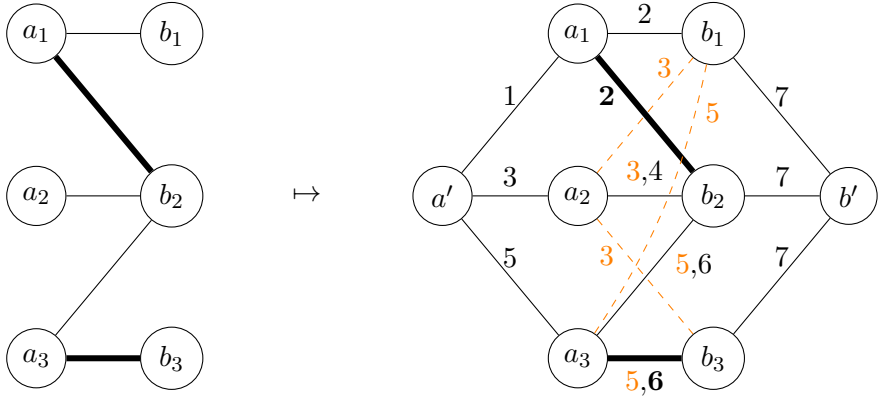
\begin{figure}[t]
	\center
	\begin{tikzpicture}[scale=.9,yscale=1.2]
		\path
			(0,4) node[circle, draw] (a1) {$a_1$}
			(0,2) node[circle, draw] (a2) {$a_2$}
			(0,0) node[circle, draw] (a3) {$a_3$}
			(2,4) node[circle, draw] (b1) {$b_1$}
			(2,2) node[circle, draw] (b2) {$b_2$}
			(2,0) node[circle, draw] (b3) {$b_3$};
		\draw (a1) edge (b1);
		\draw (a1) edge[line width=0.8mm] (b2);
		\draw (a2) edge (b2);
		\draw (a3) edge (b2);
		\draw (a3) edge[line width=0.8mm] (b3);
		
		\path (4,2) node {$\mapsto$};
		
		\path
			(6,2) node[circle, draw] (a) {$a'$}
			(8,4) node[circle, draw] (a1) {$a_1$}
			(8,2) node[circle, draw] (a2) {$a_2$}
			(8,0) node[circle, draw] (a3) {$a_3$}
			(10,4) node[circle, draw] (b1) {$b_1$}
			(10,2) node[circle, draw] (b2) {$b_2$}
			(10,0) node[circle, draw] (b3) {$b_3$}
			(12,2) node[circle, draw] (b) {$b'$};
		\draw (a) -- (a1) node[midway, above] {1};
		\draw (a) -- (a2) node[midway, above] {3};
		\draw (a) -- (a3) node[midway, above] {5};
		
		\draw (b) -- (b1) node[midway, above] {7};
		\draw (b) -- (b2) node[midway, above] {7};
		\draw (b) -- (b3) node[midway, above] {7};
		
		\draw (a1) -- (b1) node[midway, above] {2};
		\draw (a1) edge[line width=0.8mm] node[pos=.1, below=2pt] {\textbf{2}} (b2);
		\draw (a2) edge[dashed,orange] node[pos=0.7, above=2pt] {\color{orange} 3}
		(b1); \draw (a2) edge[dashed,orange] node[pos=0.1, below=2pt] {\color{orange}
		3} (b3); \draw (a2) -- (b2) node[midway, above] {{\color{orange}3},4};
		\draw (a3) edge node[pos=0.6, right=4pt] {{\color{orange}5},6} (b2);
		\draw (a3) edge[line width=0.8mm] node[midway, below] {{\color{orange}5},\textbf{6}} (b3);
		\draw (a3) edge[dashed,orange, bend right=12]  node[pos=0.85,right] {5} (b1);
	\end{tikzpicture}	
	\caption{Given a static bipartite graph $G$, we construct a temporal graph $\TG$ such that the number of matchings in~$G$ equals the number of strict paths from $a'$ to $b'$ in $\TG$. In this example, the matching $\{\{a_1,b_2\},\{a_3,b_3\}\}$ (highlighted in bold) translates to the path 
	$(\trans{a'}{1}{a_1}\tranz{2}{b_2}\tranz{5}{a_3}\tranz{6}{b_3}\tranz{7}{b'})$.
	\label{figure:gadgetmatching}
	}
\end{figure}

\begin{theorem}
\problem{Strict Temporal Paths} is \SPC.
\end{theorem}
\begin{proof}
We reduce from the \SP-complete problem \problem{Imperfect
Matchings}~\cite{valiant_complexity_1979}. Given a bipartite graph $G=(A \cup B,
E)$, we construct a temporal graph $\TG = (A \cup B \cup \{a',b'\}, \E, T)$ such
that the number of \emph{non-empty} matchings in $G$ is equal to the number of
strict temporal paths from~$a'$ to~$b'$ in $\TG$. An example for the
transformation is shown in \cref{figure:gadgetmatching} and it is easy to
check that the transformation can be computed in polynomial time.

The temporal edge set $\E$ is constructed as follows: For each edge
$\{a_i,b_j\} \in E$, we create a temporal edge~$(\{a_i,b_j\},2\cdot i)$. These
edges are meant to represent the edges of the original graph and will be called
\emph{forward-edges}.
The vertices $a_i \in A$ are connected to~$a'$ at time step $2\cdot i-1$.
All~$b_j \in B$ are connected to $b$ at the last time step~$T$. For~$i>1$ we
connect each~$a_i$ to each $b_j$ at time step $2\cdot i -1$; these edges are
drawn dashed in orange in \cref{figure:gadgetmatching} and we will refer to them
as \emph{back-edges}. 

We justify the terms \emph{forward-edge} and \emph{back-edge} by showing that
for any temporal path from~$a'$ to $b'$, every transition with an even time
label goes from an $a\in A$ to a $b \in B$ (\emph{forward}) and exactly the
other way from a $b\in B$ to an $a \in A$ (\emph{back}) for every transition
with an odd time label. By construction, each vertex $a_i\in A$ is incident to time edges with at most two time
labels:~$2\cdot i - 1$ and~$2\cdot i$. Hence, any temporal path containing a
transition $\trans{b_j}{2\cdot i}{a_i}$ ends in $a_i$ since no time edge with a
higher time label will be available. By an anologous argument, back-edges
cannot be used forwardly because it is impossible to arrive in $a_i$ before time
$2\cdot i-1$.

As a consequence, on any temporal path from $a'$ to $b'$, every back-edge is
followed by a forward-edge and every forward-edge is followed either by a back-edge or by the
final edge to $b'$. Thus, for any matching $M =
\{\{a_{i_1},b_{j_1}\},\dots,$ $\{a_{i_m},b_{j_m}\}\}$ of size $m \in \NN^+$
there is exactly one temporal path from $a'$ to $b'$ containing exactly the
forward edges corresponding to $M$, and conversely, for each temporal path $P$
from $a'$ to $b'$ there is exactly one matching corresponding to the forward-edges in~$P$. Thus, the
number of non-empty matchings in~$G$ equals the number of temporal paths
from~$a'$ to~$b'$ in~$\TG$ which implies that we have a polynomial-time
counting reduction.
\end{proof}

Analogously to the case of non-strict paths, the \SP-hardness of \problem{Strict temporal paths} implies the \SP-hardness of \problem{Strict Foremost Paths} and \problem{Strict Fastest Paths}.

\begin{corollary}\label{corollary:hardness foremost}
\problem{Strict Foremost Paths} and \problem{Strict Fastest Paths} are \SPC.
\end{corollary}

We have shown that counting strict and non-strict temporal paths is \SPC. This allows us to prove this section's main result. %We defer the proof to \cref{app:proof}.

\begin{theorem}\label{thm:hardbetweenness}
\problem{Temporal Betweenness} based on foremost or fastest, strict or non-strict paths is \SPC.
\end{theorem}
\begin{proof}
We prove the \SP-hardness by a polynomial-time counting reduction from
\problem{(Strict) Temporal Paths}.
Let $\TG=(V,\E,T)$ be a temporal graph with vertices $a$ and $b$. Let $p$ be the
number of (strict) temporal paths from~$a$ to~$b$.
We construct a temporal graph $\TG'=(V',\E',T')$ with $V' = V \cup
\{a',b',v'\}$, lifetime~$T' = T+2$, and $\E' = \{ \tedge{u}{t+1}{v} \mid
\tedge{u}{t}{v} \in \E \} \cup
\{\tedge{a'}{1}{a},\tedge{a'}{1}{v'},\tedge{v'}{T+2}{b'},\tedge{b}{T+2}{b'}\}$.
The construction is illustrated in Figure \ref{figure:gadgetbetweenness} and
can clearly be computed in polynomial time.

\begin{figure}[t]
	\center
	\begin{tikzpicture}[scale=.9]
		\path
			(2,2) node[circle, draw](1) {$a$}
			(6,2) node[circle, draw](2) {$b$}
			(-2,2) node[circle, draw](3) {$a'$}
			(10,2) node[circle, draw](4) {$b'$}
			(4,0) node[circle, draw](5) {$v'$};
		\draw [line width=.8mm, decorate,decoration={snake,amplitude=2mm,segment length=9mm}] 
			  (1) -- (2) node [midway, above = 5pt] {$[2,T+1]$};
		\draw (3) -- (1) node [midway, above = 3pt] {1};
		\draw (2) -- (4) node [midway, above = 3pt] {$T+2$};
		\path (3) edge [bend right=15] node [midway, below = 3pt] {1} (5);
		\path (5) edge [bend right=15] node [midway, below = 3pt] {$T+2$}  (4);
		\draw[rounded corners=18pt] (1,1) -- (1,3) -- (2,4) -- (4,3) -- (5.5,3.5) --
		(7,3) -- (7,1) -- (5,1.5) -- (3,0.5) -- cycle;
	\end{tikzpicture}
	\caption{Given a temporal graph $\TG$, we construct a temporal graph $\TG'$ such that we can compute the number of $a$-$b$-paths in $\TG$ from the betweenness of $v$ in $\TG'$.}
	\label{figure:gadgetbetweenness}
\end{figure}
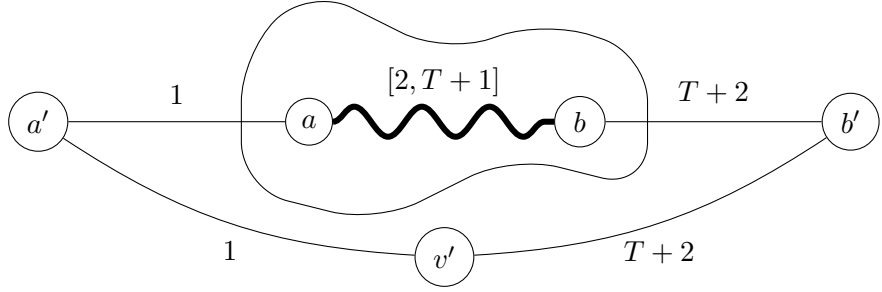

By construction, every temporal path from~$a'$ to~$b'$ is both fastest and
foremost and there are exactly four (fastest/foremost) temporal paths going
through $v'$, connecting the pairs $(a,b)$, $(a,b')$, $(a',b)$, and~$(a',b')$, respectively:
\begin{enumerate}
	\item $(\trans{a'}{1}{v'}\tranz{T+2}{b'})$,
	\item $(\trans{a}{1}{a'}\tranz{1}{v'}\tranz{T+2}{b'})$,
	\item $(\trans{a}{1}{a'}\tranz{1}{v'}\tranz{T+2}{b'}\tranz{T+2}{b})$, and
	\item $(\trans{a'}{1}{v'}\tranz{3}{b'}\tranz{T+2}{b})$.
\end{enumerate}
We observe that for each of these four pairs, there are $p+1$ foremost (and
fastest) (strict) temporal paths, and for any other pair of vertices, no
temporal path goes through~$v$ at all.
This allows us to compute~$p$ from the temporal betweenness centrality based on
foremost (fm) or fastest (fa) temporal paths of~$v$:

\begin{align*}
\bcfa(v) = \bcfm(v) &= \sum_{s \neq v \neq z} \pdforemost{sz}{v}
	  = 4\cdot\pdforemost{ab}{v} = \frac{4}{p+1}\\
	  \Rightarrow p &= \frac{4}{\bcfm(v)}-1.
\end{align*}
This implies that we have a polynomial-time
counting reduction.
\end{proof}

\section{Counting Temporal Paths and Temporal Dependencies}
\label{ch:adaptation}
In the previous section, we have shown that computing the temporal betweenness
centrality based on fastest or foremost paths is \SPC.
%, and we have identified properties of those path classes that are fundamentally different from static shortest paths, making them harder to count. 
%There are, however, other concepts of optimal paths that are more promising, including shortest paths and specialized combinations of optimality.
Now we investigate ways to extend the approach of Brandes' algorithm
\cite{brandes_faster_2001} for static graphs to variants of temporal
betweenness based on types of optimal temporal paths for which the
counting problem is not intractable. In this section, we provide the theoretical
foundations for designing efficient algorithms to compute temporal betweenness centrality.

\subsection{Counting Shortest Temporal Paths}
	\label{sec:brandes-temp-shortest}
We first show how we can count shortest temporal paths efficiently. Our results
from \cref{ch:hardness} indicate that this is not possible for the case of
foremost and fastest temporal paths.

As it turns out, what we need to count is a slightly different version of
shortest temporal paths, which we call \emph{$t$-shortest} temporal paths.
A temporal path $P$ from $s$ to $z$ is a $t$-shortest temporal path if it
arrives at time $t$ and if there is no temporal path $P'$ from $s$ to $z$ which arrives
at time $t$ and is shorter than $P$.
%To avoid confusion, we call shortest temporal paths from $s$ to $z$ \emph{all-time shortest}.
We show that similarly to the static case, $t$-shortest temporal paths $P$ have
the property that every prefix is a $t'$-shortest temporal path, where $t'$ is
the time when the prefix arrives at its last vertex. Formally, we show
the following.
\begin{lemma}%[Prefix property for shortest temporal paths]
\label{lemma:prefix shortest}
%Let $\TG=(V,\E,T)$ be a temporal graph. 
Let $P = (\trans{s}{t_1}{\dots}\tranz{t'}{v}\tranz{t''}{s'}\tranz{t'''}{\ldots}\tranz{t}{z})$ be a $t$-shortest temporal path from $s$ to $z$. Then $P' =(\trans{s}{t_1}{\dots}\tranz{t'}{v})$ is a $t'$-shortest temporal path from $s$ to $v$.
\end{lemma}
\begin{proof}
Assume that $P = (\trans{s}{t_1}{\dots}\tranz{t}{z})$ is a $t$-shortest temporal
path but $P' =(\trans{s}{t_1}{\dots}\tranz{t'}{v})$ is not a $t'$-shortest temporal path. Then there is a shorter temporal path $P'' = (\trans{s}{t_1'}{\dots}\tranz{t'}{v})$ which implies that $(P''\tranz{t''}{s'}\tranz{t'''}{\ldots}\tranz{t}{z})$ is shorter than $P$, a contradiction.
\end{proof}
\cref{lemma:prefix shortest} allows us to formulate a recursive relation of the
number of shortest temporal paths. To do this, we need to define the
predecessor of a vertex appearance~$(v,t)$ on a temporal path.

\begin{definition}\label{def:predecessor}
Let $P$ be a temporal starting at $s$ that contains the transitions
$\trans{w'}{t'}{w}\tranz{t}{v}$.
Then $(w,t')$ is the \emph{predecessor} of $(v,t)$ on~$P$. Let $P$ also contain the transition $\trans{s}{t''}{v'}$. Then $(s,0)$ is the predecessor of $(v',t'')$.

Let $s\in V$ be any vertex and~$\va{v}{t} \in V\times[T]$ any vertex
appearance. Then $P^{(\star)}_{s}(v,t)$ denotes the \emph{set of predecessors}
of $(v,t)$ on $\star$-optimal temporal paths starting at $s$.
%\begin{itemize}
%	\item $P^{\textup{(sh)}}_{s}(v,t)$ for the predecessors of $(v,t)$ on shortest paths.
%	\item $P^{\textup{(fm)}}_{s}(v,t)$ for the predecessors of $(v,t)$ on foremost paths.
%	\item $P^{\textup{(fa)}}_{s}(v,t)$ for the predecessors of $(v,t)$ on fastest paths.
%\end{itemize}
\end{definition}

If we want to use the predecessor relation to formulate recursive function, then
we need that the relation is acyclic, otherwise we are not guaranteed to reach
a base case.
To show this formally, we define the directed \emph{predecessor graph} for $\star$-optimal temporal paths starting at some vertex $s\in V$ as follows.
Given a temporal graph $\TG$, the vertex set of the predecessor graph is the set of \emph{vertex appearances} in $\TG$. The arc set is given by the ordered pairs of vertex appearances such that there is a $\star$-optimal temporal path starting at $s$ that arrives in these vertex appearances in that order, that is, we add an arc from a vertex appearance $(v,t)$ to another vertex appearance~$(w,t')$ if $(v,t)\in P^{(\star)}_{s}(w,t')$.
\begin{definition}[Predecessor Graph]
Let $\TG=(V,\E,T)$ be a temporal graph. Let $s \in V$. The \emph{predecessor graph} of $\TG$ is the directed static graph given by
\[ 
	G^{(\star)}_{\operatorname{pre}}(s) := (V \times ([T]\cup\{0\}), E),
\]
where 
\begin{align*}
	E := \{((v,t),(w,t')) \mid (v,t)\in P^{(\star)}_{s}(w,t')\}.
\end{align*}
\end{definition}
It is easy to observe that the predecessor graph is acyclic for all strict optimal temporal path versions, since the time stamps of the vertex appearances are strictly increasing along the arcs of the graph. 
%An example for the construction for non-strict shortest temporal paths is shown in \cref{figure:pathgraph}. In the example, the predecessor graph is acyclic. 
In the case of $t$-shortest temporal paths, we show next that this also holds in
the non-strict case.
\begin{lemma}%[Acyclic shortest paths]
\label{lemma:acyclictshortest}
Let $\TG=(V,\E,T)$ be a temporal graph. The predecessor graph $G^{\textup{(n-}t\textup{-sh)}}_{\operatorname{pre}}(s)$ for non-strict $t$-shortest (n-$t$-sh) temporal paths is acyclic for all $s\in V$ and all $t\in T$.
\end{lemma}
\begin{proof} Proof by contradiction. 
Let $\TG=(V,\E,T)$ be a temporal graph with source $s$. 
Let~$\dist:V\times[T] \rightarrow \NN$ be the function that returns the length of the shortest path in the predecessor graph $G^{\textup{(n-}t\textup{-sh)}}_{\operatorname{pre}}(s)$ from $(s,0)$ to any vertex appearance. Note that for any vertex appearance $(v,t')$ we have that $\dist(v,t')$ is also the length of a $t'$-shortest temporal path from $s$ to $v$. Hence, for any vertex appearances $(v,t')$ and $(w,t'')$, if the predecessor graph contains the arc $((v,t'),(w,t''))$, then $\dist(v,t') < \dist(w,t'')$.
Assume that $G^{\textup{(n-}t\textup{-sh)}}_{\operatorname{pre}}(s)$ contains a cycle 
$C = ((v,t'),\dots,(v,t'))$. 
Then $\dist(v,t') < \dist(v,t')$, a contradiction.
\end{proof}
By an analogous argument, we get the following, which we need in the next subsection.
\begin{lemma}%[Acyclic shortest paths]
\label{lemma:acyclicshortest}
Let $\TG=(V,\E,T)$ be a temporal graph. The predecessor graph $G^{\textup{(n-sh)}}_{\operatorname{pre}}(s)$ for non-strict shortest temporal paths is acyclic for all $s\in V$.
\end{lemma}

With the help of this notation and \cref{lemma:prefix shortest} (and
\cref{lemma:acyclictshortest} for the non-strict case), we can formulate the following recursion for the number of $t$-shortest temporal paths.
\begin{corollary}[Counting of $t$-Shortest Temporal Paths]
	\label{lemma:t-shortest counting}
Let $s$ be a source and let~$z$ be a vertex with $s \neq z$. The number of $t$-shortest temporal paths ($t$-sh)  from $s$ to $z$ is given by:
	\begin{align*}
		\tshortest{t}{sz} = \sum_{\va{v}{t'}\in 
		\predecessorstshortest{t}{s}{z,t}} \tshortest{t'}{sv}.
	\end{align*}
\end{corollary}

\subsection{Temporal Dependencies}
%By definition, the betweenness centrality of a vertex $v$ depends on the number of shortest paths between pairs of vertices which pass through $v$. 
Similarly to Brandes \cite{brandes_faster_2001} we introduce so-called
\emph{dependencies} to obtain an alternative formulation for the betweenness
which we then use for the actual computation.
Recall the definition of (static) betweenness centrality \[ C_B(v) = \sum_{s,
z\in V} \frac{\sigma_{sz}(v)}{\sigma_{sz}},\] where $\sigma_{sz}$ is the number
of all shortest paths from $s$ to $z$ and $\sigma_{sz}(v)$ is the fraction of
these paths which pass through $v$. Brandes \cite{brandes_faster_2001} calls
the latter the \emph{pair-dependency} of $s$ and $z$ on $v$. Brandes' central
observation is that the betweenness centrality can be computed faster by first
computing partial sums of the form \[ \delta_{s\bullet}(v) = \sum_{z \in V}
\frac{\sigma_{sz}(v)}{\sigma_{sz}},\] the \emph{dependency} of $s$ on
$v$.
%\cref{figure:dependencies} shows an example graph to illustrate the concepts of dependencies and pair-dependencies. 
He showed that the dependencies obey a recursive relation which can be exploited algorithmically. 
%Hence, a dynamic program can compute the dependency of a vertex without the need to explicitly compute all pair-dependencies on that vertex. This reduces the running time on sparse graphs.

Both concepts, dependencies and pair-dependencies, naturally generalize to optimal paths in temporal graphs.
In some cases, we may be interested in the dependency on a vertex at a specific time instead of the whole lifetime. Hence, we introduce \emph{temporal dependencies} of vertex appearances:

\begin{definition}[Temporal Dependency] Let $s$ be a source and let $z$ be a target vertex. Let~$v \neq s$ and $v \neq z$. For any $t\in [T]$, the temporal \emph{pair-dependency} of $s$ and $z$ on $(v,t)$ and the \emph{dependency} of $s$ on $(v,t)$ are given by:
%	\[
	\begin{align*}
		&\pdgeneric{sz}{v,t} := \begin{cases}
    0, & \text{if } \sigma^{(\star)}_{sz} = 0 \\
    \frac{\sigma^{(\star)}_{sz}(v,t)}{\sigma^{(\star)}_{sz}}, &
    \text{otherwise,} \end{cases} \\
		&\dgeneric{s}{v,t} := \sum_{z\in V} \delta^{(\star)}_{sz}(v,t).
	\end{align*} %\star \in \{\text{sh},\text{fa},\text{fm}\}.
%	\]
\end{definition}

From our definition of the modified temporal betweenness of a vertex
appearance it follows that we can use the temporal
dependencies to compute it.
\begin{observation}\label{lemma:depbet}
For any vertex appearance~$\va{v}{t} \in V\times[T]$ it holds:
\[
\hat{C}^{(\star)}_B(v,t) = \sum_{s\in V} \dgeneric{s}{v,t}.
\]
\end{observation}

Now we state our central result allowing us to design algorithms to compute the temporal betweenness centrality based on shortest temporal paths.
\begin{lemma}[Temporal Dependency Accumulation]
	\label{lemma:shortest-dependency}
Fix a source $s\in V$. For any vertex appearance~$\va{v}{t} \in V\times[T]$ it holds:
\begin{align*}
	\dshortest{s}{v,t} &= \pdshortest{sv}{v,t} + \sum_{\va{w}{t'}:\va{v}{t}\in \predecessorsshortest{s}{w,t'}}
		\frac{\tshortest{t}{sv}}{\tshortest{t'}{sw}}\cdot \dshortest{s}{w,t'}. 
\end{align*}
\end{lemma}
\begin{proof}
We start by rewriting the pair-dependency of a vertex appearance as the summed
fractions of shortest temporal paths using a specific transition from $v$ to
some vertex~$w$ at some time step:
\begin{align*}
	 \dshortest{s}{v,t}
	& = \sum_{z\in V} \pdshortest{sz}{v,t} \\
	& = \pdshortest{sv}{v,t} + \sum_{z\in V} \sum_{\va{w}{t'}:\va{v}{t}
		\in \predecessorsshortest{s}{w,t'}} \pdshortest{sz}{(v,t),(\{v,w\},t')},
\end{align*}
where $\pdshortest{sz}{(v,t),(\{v,w\},t')}$ is the fraction of shortest temporal paths from vertex~$s$ to~$z$ that use the transitions $\trans{v'}{t}{v}\tranz{t'}{w}$ for some~$v'$.
Note that for the case that $z=v$ there is no vertex appearance $(w,t')$ such that $\va{v}{t}
		\in \predecessorsshortest{s}{w,t'}$, hence we need $\pdshortest{sv}{v,t}$ as an additional summand.
Furthermore, we have that
\[
\pdshortest{sz}{(v,t),(\{v,w\},t')}
	= \frac{\tshortest{t}{sv}}{\tshortest{t'}{sw}} \cdot
	  \frac{\shortest{sz}(w,t')}{\shortest{sz}}.
	 \]
Inserting these cases into the double sum above yields the following:
\begin{align*}
& \sum_{z\in V} \sum_{\va{w}{t'}:\va{v}{t}
		\in \predecessorsshortest{s}{w,t'}} \pdshortest{sz}{(v,t),(\{v,w\},t')}\\
=& \sum_{\va{w}{t'}:\va{v}{t} \in \predecessorsshortest{s}{w,t'}} 
	\sum_{z \in V} \pdshortest{sz}{(v,t),(\{v,w\},t')}\\
=& \sum_{\va{w}{t'}:\va{v}{t} \in \predecessorsshortest{s}{w,t'}}
	\left(\frac{\tshortest{t}{sv}}{\tshortest{t'}{sw}} \cdot \sum_{z\in V}
	  \frac{\shortest{sz}(w,t')}{\shortest{sz}}\right)\\
=& \sum_{\va{w}{t'}:\va{v}{t} \in \predecessorsshortest{s}{w,t'}}
		\frac{\tshortest{t}{sv}}{\tshortest{t'}{sw}}\cdot \dshortest{s}{w,t'}.
\end{align*}
The lemma statement now follows immediately.
\end{proof}

\cref{lem:betversions} together with \cref{lemma:depbet} show that we can use
temporal dependencies to compute temporal betweenness centrality.
\cref{lemma:shortest-dependency} shows that the temporal dependencies obey a
recursive relation and \cref{lemma:acyclicshortest} shows that it is acyclic.
This together with \cref{lemma:t-shortest counting} allows us to design a
polynomial-time algorithm to compute temporal betweenness centrality based on
shortest temporal paths. Before we give the specific algorithm description and
running time analysis, we describe some other temporal betweenness concepts
that can be computed in a similar fashion.

\subsection{Specialized Optimality Concepts}
	\label{sec:acyclic}
So far, our results allow us to efficiently compute the temporal betweenness
for strict and non-strict \emph{shortest} temporal paths. For all other of the
six canonical variants we have shown computational hardness results in
\cref{ch:hardness}.
In the following, we show that for some specialized optimality concepts for
temporal paths, we can also efficiently compute the corresponding temporal
betweenness.
Specifically, we consider \emph{shortest foremost} temporal paths and \emph{prefix-foremost} temporal paths. 

\paragraph{Shortest foremost temporal paths.}
A shortest foremost path is a foremost temporal path that is not longer than
any other foremost temporal path. It may not be a shortest temporal path
overall, however, as even shorter temporal paths with higher arrival times may
exist. Shortest foremost temporal paths can be regarded as temporal paths that
prioritize a low arrival time and use a low edge count as a tie-breaker.
\begin{definition}
Let $s,z\in V$. A temporal path~$P$ from~$s$ to~$z$ is a \emph{shortest
foremost} temporal path if~$P$ is a foremost temporal path from $s$ to $z$ and
for all foremost temporal paths~$P'$ from~$s$ to~$z$ it holds: $\abs{P} \le
\abs{P'}$.
\end{definition}

We observe that a shortest foremost temporal path is a $t$-shortest temporal path for the earliest possible arrival time. Hence, our findings in the previous subsection translate very easily to shortest foremost temporal paths. %In fact, the algorithm for shortest paths discussed in \cref{ch:algorithms} computes the shortest foremost paths as well.

\paragraph{Prefix foremost temporal paths.}
%In the example given in \cref{figure:hamilton}, every path from $s$ to $z$ is foremost because $z$ can only be reached very late. On an intuitive level, this means that a path may waste a lot of time early on without any impact on the final arrival time because there is a bottleneck immediately before the target. In this section, we consider the paths that do not show this behavior. 
Now we consider the class of foremost temporal paths for which every prefix of
the temporal path is foremost as well. That is, every vertex that is visited by
such a temporal path is visited as soon as possible. These temporal paths are
called \emph{prefix foremost paths}.
% and consider a variation of betweenness
%centrality based on prefix-foremost paths.

\begin{definition}
Let $s,z\in V$. A path $P = (e_1,\dots,e_k)$ from $s$ to $z$ is a \emph{prefix
foremost} temporal path if $P$ is a foremost path and every prefix $P' =
(e_1,\dots,e_{k'})$, $k'\le k$ is a foremost temporal path as well.
\end{definition}

Wu et al.~\cite{wu_efficient_2016} showed that there is always at least one
prefix foremost path between any pair of vertices $s$ and $z$ unless $z$ is not reachable from $s$ at all.
We can observe that in the straightforward reduction from \problem{Paths} that
yields \cref{prop:temppaths} we also get that counting non-strict prefix
foremost temporal paths is computationally hard.
\begin{proposition}
\problem{Prefix Foremost Paths} is \SPC.
\end{proposition}

The main difference between strict and \emph{non-strict} prefix foremost paths
foremost paths is that in the former case, the predecessor relation is acyclic,
and in the latter it is not. Intuitively, this is the reason why we get hardness
in the non-strict case and tractability for the strict version.

For strict prefix foremost paths, one can show that the counting problem and
the computation of the corresponding betweenness are solvable in polynomial
time by very similar arguments as for shortest temporal paths. The main difference is that the
time stamp of a predecessor is unique, hence we can define the set of
predecessors as a set of vertices as opposed to a set of vertex appearances,
otherwise the arguments are analogous to the ones used in the proof of
\cref{lemma:shortest-dependency}.
This allows us to simplify the recursive function for the temporal dependency
for prefix foremost temporal paths which also results in a faster computation.
%
%Let $\sigmaprefix{sz}$ denote the number of strict prefix-foremost paths from $s$ to $z$.
%
\begin{lemma}[Strict Prefix Foremost Dependency Accumulation]
	\label{lemma:prefix-dependency}
Fix a source $s\in V$. For any vertex~$v \in V, v \neq s$ it holds:
\begin{align*}
	& \dprefix{s}{v} = 1 + \sum_{w:v\in \predecessorsprefix{s}{w}}
		\frac{\sigmaprefix{sv}}{\sigmaprefix{sw}}\cdot \dprefix{s}{w}. &\\
\end{align*}
\end{lemma}
\begin{proof}
We start by rewriting the pair-dependency of a vertex appearance as the summed
fractions of prefix foremost temporal paths using a specific transition from $v$
to some vertex~$w$ at some time step:
%We start by rewriting the pair-dependency of a vertex as the summed number of
% path-fractions using a specific transition from $v$ to some vertex~$w$:
\begin{align*}
	 \dprefix{s}{v}
	&= \sum_{z\in V} \pdprefix{sz}{v} \\
	&= \pdprefix{sv}{v} + \sum_{z\in V} \sum_{w:v
		\in \predecessorsprefix{s}{w}} \pdprefix{sz}{v,\{v,w\}},
\end{align*}
where $\pdprefix{sz}{v,\{v,w\}}$ is the fraction of prefix foremost temporal
paths from vertex~$s$ to~$z$ that use the transitions
$\trans{v}{t}{w}$ for some~$t$.
Note that for the case that $z=v$ there is no vertex $w$ such that $v
		\in \predecessorsprefix{s}{w}$, hence we need $\pdprefix{sv}{v}$ as an additional summand.
Furthermore, we have that
\[
\pdprefix{sz}{v,\{v,w\}}
	= \frac{\sigmaprefix{sv}}{\sigmaprefix{sw}} \cdot
	  \frac{\sigmaprefix{sz}(w)}{\sigmaprefix{sz}}.
	 \]
Inserting these cases into the double sum above yields the following:
\begin{align*}
& \sum_{z\in V} \sum_{w:v
		\in \predecessorsprefix{s}{w}} \pdprefix{sz}{v,\{v,w\}}\\
=& \sum_{w:v
		\in \predecessorsprefix{s}{w}} 
	\sum_{z \in V} \pdprefix{sz}{v,\{v,w\}}\\
=& \sum_{w:v
		\in \predecessorsprefix{s}{w}}
	\left(\frac{\sigmaprefix{sv}}{\sigmaprefix{sw}} \cdot \sum_{z\in V}
	  \frac{\sigmaprefix{sz}(w)}{\sigmaprefix{sz}}\right)\\
=& \sum_{w:v\in \predecessorsprefix{s}{w}}
		\frac{\sigmaprefix{sv}}{\sigmaprefix{sw}}\cdot \dprefix{s}{w}.
\end{align*}
The lemma statement now follows immediately.
\end{proof}

\section{Algorithms and Running Time Analysis}
\label{ch:algorithms}
%We have shown that shortest, shortest foremost, and strict prefix-foremost paths have sufficiently convenient properties to allow dynamic counting and dependency accumulation in a similar way as done by Brandes' algorithm \cite{brandes_faster_2001} for the case of static graphs. We use our findings to give algorithms for the three betweenness variants discussed in \cref{ch:adaptation}, starting with \problem{Shortest Temporal Betweenness}.
In the following we give detailed descriptions of the algorithms we use to
compute the temporal betweenness variants discussed in \cref{ch:adaptation}.
We first describe how to adapt Brandes' algorithm~\cite{brandes_faster_2001} to
the temporal setting and then describe the competing approach of computing
temporal betweenness using so-called \emph{static expansions}.

\subsection{Algorithms Based on Adaptations of Brandes' Algorithm to the
Temporal Setting} Brandes' algorithm~\cite{brandes_faster_2001} (for static
graphs without edge weights) is based on breadth-first search.
For each vertex, the \problem{Single-Source Shortest Paths} problem is solved
once. At the end of each iteration, the dependency of the respective+$s$ on all
other vertices~$v$ is added to the betweenness score of $s$.

Our algorithms for temporal betweenness will follow roughly the same structure. Instead of breadth-first search, the appropriate algorithm for the respective temporal path class is used as the base for the algorithm. To describe our algorithm formally, we need the notion of temporal neighborhoods.
\begin{definition}[Temporal Neighborhood]
For a temporal graph~$\TG=(V,\E,T)$ and a vertex~$v\in V(\TG)$, we call $N(v) := \{(u,t) \in V(\TG)\times [T] \mid \{u,v\} \in E_t(\TG)\}$ the \emph{temporal neighborhood} of $v$.
\end{definition}

\paragraph{Shortest (foremost) temporal betweenness.}
We have shown that the temporal dependencies for shortest (foremost) paths follow a similar recursion as for static graphs. 

\cref{alg:shortest} uses a $\abs{V} \times T$-table to store the number of shortest temporal paths to all vertex appearances. The overall structure of the algorithm is similar to Brandes' algorithm~\cite{brandes_faster_2001}---a single-source-all-shortest-paths traversal from each vertex to all vertex appearances is performed and the count of shortest temporal paths is stored in the aforementioned table. At the end of each iteration, we add the dependencies found for each vertex appearance to the betweenness score of the respective vertex. 

In order to count shortest foremost temporal paths instead of shortest temporal paths, only the dependency accumulation needs to be changed: instead of checking whether the paths have minimal length, the algorithm checks for minimal arrival time. 

\begin{proposition}
Given a temporal graph $\TG=(V,\E,T)$, \cref{alg:shortest} computes the shortest
temporal betweenness and the shortest foremost temporal betweenness of all $v\in
V$ in $\bigO(|V|^3 \cdot T^2)$ time and $\bigO(|V|\cdot T+|\E|)$ space.
\end{proposition}
\begin{proof}
The correctness of \cref{alg:shortest} follows from \cref{lem:betversions}
together with \cref{lemma:depbet}, which show that we can use temporal
dependencies to compute the temporal betweenness, and from
\cref{lemma:t-shortest counting} and \cref{lemma:shortest-dependency} (with
\cref{lemma:acyclicshortest}) since our algorithm dynamically computes the
values given by the recursive formulas stated in the lemmas.

It is easy to check that the running time of \cref{alg:shortest} is
upper-bounded by $|V|$ (the outer loop in \cref{line:outerloop}) times the
total number of elements in $Q$ (the loop from \cref{line:qloop}) times the
maximum size of a temporal neightborhood (the loop from \cref{line:nloop}).
Both the total number of elements in~$Q$ and the maximum size of a temporal
neightborhood is in $\bigO(|V|\cdot T)$. The claimed running time follows.

The space consumption of \cref{alg:shortest} is in $\bigO(|V|\cdot T+|\E|)$,
since we need to store matrices of size $|V|\cdot T$ for the temporal path counts and the temporal dependencies.
\end{proof}

\begin{algorithm}[t!]
\caption{Shortest (foremost) betweenness in temporal graphs}\label{alg:shortest}
\footnotesize
\begin{algorithmic}[1]
\Input{A temporal graph $\TG = (V,\E,T)$.}
\Output{Betweenness $\bcsh$ and $\bcfm$ of all vertices $v \in V(\TG)$}
	\For{$s \in V$}\label{line:outerloop}
		\For{$v \in V$}\Comment{Initialization}
			\State $\dist[v] \gets -1$; $\sigma[v] \gets 0$; $\tmin[v] \gets -1$
		\EndFor
		\For{$(\{u, v\}, t) \in \E$}
			\State $\delta^{\sh}[v, t] \gets 0$; $\delta^{\fm}[v, t] \gets 0$
			\State $\sigma[v, t] \gets 0$; $P[v, t] \gets \emptyset$; $\dist[v, t] \gets -1$
		\EndFor
		\State $\dist[s] \gets 0$; $\dist[s, 0] \gets 0$; $\tmin[s] \gets 0$
		\State $\sigma[s] \gets 1$; $\sigma[s, 0] \gets 1$
		\State $S \gets$ empty stack; $Q \gets$ empty queue; $Q \gets \operatorname{enqueue}(s,0)$
		\While{$Q$ not empty}\label{line:qloop}
		\State $(v,t) \gets \operatorname{dequeue}(Q)$
			\For{$(w,t') \in N(v)$ with $t < t'$} \Comment{($t \le t'$) for n-str}\label{line:nloop}
				\If {$\dist[w,t'] = -1$}\Comment{First arrival at $(w, t')$}
					\State $\dist[w,t'] \gets \dist[v,t] + 1$
					\If{$\dist[w] = -1$} \Comment{Shortest path to $w$}
						\State $\dist[w] \gets \dist[v,t] + 1$
					\EndIf
				\State $S \gets \operatorname{push}(w,t')$; $Q.\operatorname{enqueue}(w,t')$
				\EndIf
				\If {$\dist[w,t'] = \dist[v,t] + 1$} \Comment{Shortest path to $(w,t')$ via
				$(v,t)$} \State $\sigma[w,t'] \gets \sigma[w,t'] + \sigma[v,t]$
					\State $P[w,t'] \gets P[w,t'] \cup \{(v,t)\}$
					\If {$\dist[w, t'] = \dist[w]$} \Comment{Shortest path to $w$ via
				$(v,t)$}
						\State $\sigma[w] \gets \sigma[w] + \sigma[v,t]$
					\EndIf
				\EndIf
				\If {$\tmin[w] = -1 \textbf{ or } t' < \tmin[w]$}  \Comment{Foremost
				shortest path to $w$} \State $\tmin[w] \gets t'$
				\EndIf
			\EndFor
		\EndWhile
		\State $\bcsh[s] \gets \bcsh[s] - \lvert\{i \mid \dist[i] \geq 0\}\rvert + 1$
		\State $\bcfm[s] \gets \bcfm[s] - \lvert\{i \mid \dist[i] \geq 0\}\rvert + 1$
		\While{$(w,t') \gets \operatorname{pop}(S)$} \Comment{Vertex appearances in
		order of non-increasing distance from $s$}
			\If{$\dist[w, t'] = \dist[w]$} \Comment{Shortest path to $w$}
				\State $\delta^{\sh}[w, t'] \gets \delta^{\sh}[w, t'] + \frac{\sigma[w, t']}{\sigma[w]}$
			\EndIf
			\If{$t' = \tmin[w] $} \Comment{Foremost
				shortest path to $w$}
				\State $\delta^{\fm}[w, t'] \gets \delta^{\fm}[w, t'] + 1$
			\EndIf
			\For{$(v,t) \in P[w,t']$}
				\State $\delta^{\sh}[v,t] \gets \delta^{\sh}[v,t] 
			+ \frac{\sigma[v,t]}{\sigma[w,t']}\cdot \delta^{\sh}[w,t']$
				\State $\bcsh[v] \gets \bcsh[v] + \frac{\sigma[v,t]}{\sigma[w,t']}\cdot \delta^{\sh}[w,t']$
				\State $\delta^{\fm}[v,t] \gets \delta^{\fm}[v,t]
			+ \frac{\sigma[v,t]}{\sigma[w,t']}\cdot \delta^{\fm}[w,t']$
				\State $\bcfm[v] \gets \bcfm[v] + \frac{\sigma[v,t]}{\sigma[w,t']}\cdot \delta^{\fm}[w,t']$
			\EndFor
		\EndWhile
	\EndFor
	\State \textbf{return} $\bcsh$, $\bcfm$
\end{algorithmic}
\end{algorithm}

\paragraph{Strict prefix foremost temporal betweenness.}
\cref{alg:prefix-foremost-prioqueue} modifies Brandes'
algorithm~\cite{brandes_faster_2001} to count strict prefix foremost temporal
paths. The overall structure of the algorithm remains the same. Instead of
iterating over all neighbors, however, we add add all temporal edges of a
vertex into a priority queue (prioritizing early time labels). This allows us
to traverse the graph in a time-respecting manner and find the prefix foremost temporal paths. %The dependency-accumulation in lines 23---28 is analogous to Brandes.

Notably, since the time labels of predecessors are unique in prefix foremost
temporal paths, we do not have to consider vertex appearances in this case.
Intuitively, this is the main reason why we can achieve both a faster running
time a lower space consumption.

\begin{proposition}
Given a temporal graph $\TG=(V,\E,T)$, \cref{alg:prefix-foremost-prioqueue}
computes the prefix-foremost temporal betweenness of all $v\in V$ in $\bigO(|V|
\cdot |\E| \cdot \log |\E|)$ time and $\bigO(|\E|+|V|)$ space.
\end{proposition}
\begin{proof}
The proof is essentially analogous to the correctness proof for
\cref{alg:shortest} but uses \cref{lemma:prefix-dependency}.

It is easy to check that the running time of
\cref{alg:prefix-foremost-prioqueue} is upper-bounded by $|V|$ (the outer loop
in \cref{line:outpre}) times the size of $Q$ (the loop from \cref{line:qloop2})
times the time necessary to dequeue (\cref{line:deq}). The
size of $Q$ is in $\bigO(|\E|)$ and the time to dequeue is in $\bigO(\log
|\E|)$.
Here, we assume that $|V|\le|\E|$. The claimed running time follows.

The space consumption of \cref{alg:prefix-foremost-prioqueue} is in
$\bigO(|\E|+|V|)$ the arrays for the temporal path counts and the temporal
dependencies used in the algorithm are all of size $|V|$.
\end{proof}

\begin{algorithm}[t!]
\caption{Strict prefix-foremost betweenness}\label{alg:prefix-foremost-prioqueue}
\footnotesize
\begin{algorithmic}[1]
\Input{A temporal graph $\TG = (V,\E,T)$.}
\Output{Betweenness $\bcpfm$ of all vertices $v \in V(\TG)$}
	\For{$s \in V$}\label{line:outpre}
		\For{$v \in V$}\Comment{Initialization}
			\State $P[v] \gets \emptyset$; $\tmin[v] \gets -1$
			\State $\sigma[v] \gets 0$; $\delta[v] \gets 1$
		\EndFor
		\State $\tmin[s] \gets 0$; $\sigma[s] \gets 1$
		\State $S \gets$ empty stack
		\State $Q \gets$ empty priority queue \Comment{Prioritized by time label}
		\State $Q \gets \operatorname{enqueueAll}(\{\trans{s}{t}{v} \mid \tedge{s}{t}{v} \in \E\})$
		\While{$Q$ not empty}\label{line:qloop2}
			\State $\trans{v}{t}{w} \gets \operatorname{dequeue}(Q)$\label{line:deq}
			\If {$\tmin[w] = -1$}\Comment{First and foremost arrival in $w$}
				\State $\tmin[w] \gets t$
				\State $S \gets \operatorname{push}(w)$
				\State $Q.\operatorname{enqueueAll}(\{\trans{w}{t'}{x} 
					\mid \tedge{w}{t'}{x} \in \E, t < t'\})$
			\EndIf
			\If {$\tmin[w] = t$} \Comment{Prefix foremost path to $w$ via
				$v$}
				\State $\sigma[w] \gets \sigma[w] + \sigma[v]$
				\State $P[w] \gets P[w] \cup \{v\}$
			\EndIf
		\EndWhile
		\State $\bcpfm[s] \gets \bcpfm[s] - \lvert\{i \mid \dist[i] \geq 0\}\rvert + 1$
		\While{$w \gets \operatorname{pop}(S)$} \Comment{Vertices in
		order of non-increasing distance from $s$}
			\For{$v \in P[w]$}
				\State $\delta[v] \gets \delta[v] + \frac{\sigma[v]}{\sigma[w]}\cdot \delta[w]$
				\State $\bcpfm[v] \gets \bcpfm[v] + \frac{\sigma[v]}{\sigma[w]}\cdot \delta[w]$
			\EndFor
		\EndWhile
	\EndFor
	\State \textbf{return} $\bcpfm$
\end{algorithmic}
\end{algorithm}

\subsection{Algorithms Based on the Static Expansion}\label{sec:staticexp}
In this section, we discuss the presumably more straightforward way of computing
temporal betweenness, namely through the use of so-called \emph{static
expansions} or time-expanded temporal
graphs,
which are
a key
tool in
temporal
graph
algorithmics~\cite{Zsc+20,Akr+19,kempe_connectivity_2002,MMS19,wu_efficient_2016}.
Static expansions are directed graphs that have a vertex for every vertex
\emph{appearance} of the corresponding temporal graph and they preserve the
connectivity between vertex appearances of the temporal graph.

The main idea is to use Brandes' algorithm~\cite{brandes_faster_2001} for
(directed) static graphs on the static expansion. Of course we need to make sure
that each shortest path in the static expansion corresponds to exactly one
optimal temporal path in the temporal graph for the optimality criterion that
we are interested in.
To do this, it is necessary to exclude some paths from the betweenness
computation. To this end, we first give a generalized version of static
betweenness.
Then we present the details on how to construct static expansions for the
optimality concepts used in our work.
Notably, we do not present a static expansion for the strict prefix-foremost
case, since for this case our algorithm is significantly faster than for the
other optimality concepts.

	Let~$S\subseteq V$ be a set of ``source'' vertices and~$Z\colon S\rightarrow
	2^V$ be a function that, for any~$s\in S$, returns the set of ``terminal''
	vertices for that source.
	Now we define the~$S$-$Z$-Betweenness of a vertex~$v\in V$ to
	be
	\[
	\betweennesssz(v) = \sum_{s\in S\setminus\{v\}}\sum_{z\in
	Z(s)\setminus\{v\}}\delta_{sz}(v).
	\]
	The problem of computing the~$S$-$Z$-Betweenness for each vertex in an
	arbitrary (weighted) directed graph can be solved with a straightforwardly
	modified version Brandes' algorithm~\cite{brandes_faster_2001}. We omit the details here.
	
	\paragraph{Static expansion for shortest
	temporal betweenness.}\label{sec:shortest} We now describe how to construct a
	static expansion of a temporal graph in order to compute (strict) shortest temporal betweenness.
	 Let~$\TG=(V,\E,T)$ be
	a temporal graph.
	We then define the static expansion graph for the (strict) shortest betweenness
	as the directed unweighted graph~$\staticg = \statictuple$, with the following sets of vertices and edges:
	
	For each vertex~$v\in V$ of the original graph, we define~$T + 2$ vertices in
	the static expansion---one for each appearance and two special vertices: a
	``source'' and a ``sink.'' Formally,
	\begin{equation*}
		\staticv := \bigcup_{v\in V}\left\{v_t \mid t\in[T+1]\cup\{0\}\right\}.
	\end{equation*}
	
	For the edges we have three cases. The vertices~$v_{T+1}$ are used as sentinel
	``terminal'' vertices that will signify the end of a path in our betweenness
	computation. Hence, they do not have any outgoing edges. Now, for any temporal
	edge~$(\{v, w\}, t)\in\E$ and every~$v_{t'}$ with~$t'\leq t$, we either add the
	directed edge~$(v_{t'}, w_t)$, or the edge~$(v_{t'}, w_{t+1})$---the former in
	the non-strict case, the latter in the strict case. Moreover, we also add an
	edge to the corresponding terminal vertex, that is, we add the edge~$(v_{t'},
	w_{T+1})$. Naturally, we also go through the same process with~$w$ as we did
	with~$v$.
	Formally, we have in the non-strict case:
	\begin{equation*}
		\statice := \bigcup_{(\{v, w\}, t)\in\E}\bigcup_{0\leq t' \leq t}\{(v_{t'},
		w_t), (v_{t'}, w_{T+1}), (w_{t'}, v_t), (w_{t'}, v_{T+1})\},
	\end{equation*}
	and in the strict case:
	\begin{equation*}
		\statice := \bigcup_{(\{v, w\}, t)\in\E}\bigcup_{0\leq t' \leq t}\{(v_{t'},
		w_{t+1}), (v_{t'}, w_{T+1}), (w_{t'}, v_{t+1}), (w_{t'}, v_{T+1})\}.
	\end{equation*}
	
	Now, to use the static expansion to compute the betweenness values in the
	temporal graph, we first compute the~$S$-$Z$-betweenness on the static
	expansion with
	\begin{align*}
		S & := \{v_0\mid v\in V\}, \\
		Z(s) & := \{v_{T+1}\mid v\in V\setminus\{s\}\}.
	\end{align*}
% 	After doing that we can relate the betweenness values of the vertices of the static expansion to those of the temporal graph:
We can now compute the temporal betweenness values using the following (where
$\star$ is the corresponding optimality concept):
	\begin{equation*}
		\bc(v) = \sum_{t\in\natinterval{T}}\betweennesssz(v_t).
	\end{equation*}
	
   \paragraph{Static expansion for shortest foremost temporal betweenness.}
	This variant of static expansion will be similar as the static expansion from
	the previous paragraph. Indeed, for a temporal graph~$\TG=(V,\E,T)$ we
	define the static expansion for (strict) shortest-foremost betweenness to be
	the directed weighted graph~$\staticg = \statictuplew$, with~$\staticv$
	and~$\statice$ defined exactly as in the static expansion for (strict) shortest
	betweenness (see previous paragraph). Additionally, we use the edge weight
	function~$\weightfn$ to ``simulate'' the foremost aspect of the paths. We let all ``internal'' edges be of equal weight, but let the ``terminal'' edges have a weight related to their timestamps. We set the weight in such a way as to make it so that any path to an ``appearance''~$v_t$ is better than any path to an ``appearance''~$v_{t'}$ for~$t < t'$. Formally, we have
	\begin{equation*}
		\omega(v_t, w_{t'}) = \begin{cases} 1, & \text{if } t' \leq T, \\ (n +
		1)\cdot(t' + 1), & \text{otherwise.}\end{cases}
	\end{equation*}
	Note that this weight function works in both the strict and the non-strict
	case. 
		
	The further procedure, that is, the computation of~$S$-$Z$-betweenness and then
	the temporal betweenness values can be done in exactly the same way as in
	for shortest temporal betweenness (see the previous paragraph).

\section{Experimental Evaluation}
\label{ch:exp}
In this section, we analyze the running time of our implementations of
\cref{alg:shortest} (computing strict and non-strict shortest and shortest
foremost temporal betweenness) and \cref{alg:prefix-foremost-prioqueue}
(computing strict prefix-foremost temporal betweenness) as well as
implementations of algorithms based on static expansions (computing strict and non-strict shortest and shortest
foremost temporal betweenness, same as \cref{alg:shortest}) on several
real-world temporal graphs and we investigate the differences between the five different betweenness concepts our algorithms can compute.

\subsection{Setup and Statistics}
We implemented\footnote{The code of our implementation is available under GNU
general public license version 3 at
\url{http://fpt.akt.tu-berlin.de/software/temporal\_betweenness/}.}
\cref{alg:shortest} and \cref{alg:prefix-foremost-prioqueue} and algorithms
using the static expansions (see \cref{sec:staticexp}) in C++ and we compiled
our source code with GCC 7.5.0.
We carried out experiments on an Intel Xeon W-2125 computer with four cores clocked at 4.0\,GHz and with 256\,GB RAM, running Linux 4.15. We did not utilize the
parallel-processing capabilities.

We used the following freely available data sets for temporal graphs:
Physical-proximity networks\footnote{Available at
\url{http://www.sociopatterns.org/datasets/}.} between high school students
(``highschool-2011'', ``highschool-2012'', ``highschool-2013''~\cite{gemmetto2014mitigation,stehle2011high,fournet2014contact}), children and teachers in a primary school (``primaryschool''~\cite{stehle2011high}),
patients and health-care workers (``hospital-ward''~\cite{vanhems2013estimating}), 
attendees of the Infectious SocioPatterns event (``infectious''~\cite{isella2011s}), and
conference attendees of ACM Hypertext 2009 (``hypertext''~\cite{isella2011s}), and
%Email communication networks of the 2016 Democratic National Committee email leak (``dnc''~\cite{KONECT17}) and
an email communication networks of the Karlsruhe Institute of Technology (KIT) (``karlsruhe''~\cite{EMT2011}), 
and one social communication network (``facebook-like''~\cite{opsahl2009clustering}).
%
%\begin{itemize}
%\item physical-proximity networks\footnote{Available at http://www.sociopatterns.org/datasets/ .} between 
%\begin{itemize}
%\item high school students (``highschool-2011'', ``highschool-2012'', ``highschool-2013''~\cite{gemmetto2014mitigation,stehle2011high,fournet2014contact}), 
%\item children and teachers in a primary school (``primaryschool''~\cite{stehle2011high}),
%\item patients and health-care workers (``hospital-ward''~\cite{vanhems2013estimating}), 
%\item attendees of the Infectious SocioPatterns event (``infectious''~\cite{isella2011s}),
%\item conference attendees of ACM Hypertext 2009 (``hypertext''~\cite{isella2011s}), 
%\end{itemize}
%\item an email communication network of the 2016 Democratic National Committee email leak (``dnc''~\cite{KONECT17}),
%\item an email communication network (``karlsruhe''~\cite{EMT2011}), 
%\item a social communication network (``facebook-like''~\cite{opsahl2009clustering}).
%\end{itemize}
We summarize some important statistics about the different data sets in \cref{tab:stats}.
\begin{table*}[t!]
\tiny
  \centering
  \caption{Statistics for the data sets used in our experiments. The lifetime
  $T$ of a graph is the difference between the largest and smallest time stamp
  on an edge in the graph. The resolution~$r$ indicates how often edges were
  measured. The last five columns state the running times in seconds of our
  implementation, where the last two correspond to the algorithms based on
  static expansions. From left to right: non-strict shortest and shortest
  foremost betweenness, strict shortest and shortest
  foremost betweenness, strict prefix-foremost betweenness, non-strict shortest and shortest
  foremost betweenness computed with static-expansion-based algorithms, strict
  shortest and shortest foremost betweenness computed with
  static-expansion-based algorithms. A~-1 indicates that
  the instance was not solved within five hours.} \makebox[\textwidth]{ \begin{filecontents*}{GraphData.csv}
Data,Vertices,Edges,Lifetime (s),Resolution (s),R1,R2,R3,AN,AS
highschool-2011,126,28560,272330,20,65.9186,65.6172,0.792614,36.32319,34.69841
highschool-2012,180,45047,729500,20,197.262,196.702,2.11343,124.4933,119.01789
highschool-2013,327,188508,363560,20,2606.83,2588.58,22.3243,2060.8658,1906.0075
primaryschool,242,125773,116900,20,898.341,894.553,8.81455,512.69621,507.96691
hospital-ward,75,32424,347500,20,97.5395,97.1868,0.433082,48.85941,46.70975
infectious,10972,415912,6946340,20,995.203,942.201,6.93528,-1,-1
hypertext,113,20818,212340,20,37.0512,36.9026,0.462423,21.42619,20.42503
karlsruhe,1870,461661,123837267,1,-1,-1,311727,-1,-1
facebook-like,1899,59835,16736181,1,1298.03,1303.44,24.8583,8249.5588,8221.5252
\end{filecontents*}
\pgfplotstabletypeset[ col sep=comma, columns={Data,Vertices, Edges,Resolution (s), Lifetime (s),R1,R2,R3,AN,AS}, %,Classical
   % Degeneracy,Degeneracy0,Degeneracy3,Degeneracy5,Degeneracy7},
   columns/Data/.style={column type=l,string type, column name={Data Set}},
   columns/Vertices/.style={column type=r, column name={\# Vtc's $n$}},
   columns/Edges/.style={column type=r,int detect, column name={\# Edges $M$}},
   columns/Lifetime (s)/.style={column type=r,int detect, column name={Lifetime $T$}},
   columns/Resolution (s)/.style={column type=r,int detect, column name={Res. $r$}},
   columns/R1/.style={column type=r,int detect, column name={N-Str.\ Sh (Fm)}},
   columns/R2/.style={column type=r,int detect, column name={Str.\ Sh (Fm)}},
   columns/R3/.style={column type=r,int detect, column name={Str.\ P Fm}},
   columns/AN/.style={column type=r,int detect, column name={Ex N-Str.\ Sh
   (Fm)}}, 
   columns/AS/.style={column type=r,int detect, column name={Ex Str.\ Sh
   (Fm)}},
   %columns/Classical Degeneracy/.style={column name={Static}},
   %columns/Degeneracy0/.style={column name={$\Delta=0$}},
   %columns/Degeneracy3/.style={column name={$\sim 5^3$}},
   %columns/Degeneracy5/.style={column name={$\sim 5^5$}},
   %columns/Degeneracy7/.style={column name={$\sim 5^7$}},
   every head row/.style={before row=\toprule}, %  & & & & & \multicolumn{5}{c}{Degeneracy} \\ \cmidrule(r){6-10}, after row=\midrule},
   every last row/.style={after row=\bottomrule},
   ]{GraphData.csv}}
  \label{tab:stats}
\end{table*}%

\subsection{Experimental Results}
We now discuss the results of our experiments. We analyzed the influence of the
temporal betweenness type on the running time, the distribution of the
temporal betweenness values, and the ranking of the ten vertices with the
hightest temporal betweenness values.
\paragraph{Running time.}
As indicated by our theoretical running time bounds (see
\cref{table:complexity}), \cref{alg:shortest} (computing shortest and foremost
shortest betweenness) is several orders of magnitudes slower than
\cref{alg:prefix-foremost-prioqueue} (computing prefix foremost betweenness).
\cref{alg:shortest} solved all instances except for (which was not solved within
the timeout of five hours) within 45 minutes (keep in mind that
\cref{alg:shortest} computes shortest and foremost shortest temporal betweenness simultaneously).
\cref{alg:prefix-foremost-prioqueue} solved all instances except ``karlsruhe''
within 30 seconds. We show the specific running times in \cref{tab:stats}.
It is noticable that on small instances, the algorithms based on static
expansions are roughly a factor of two faster. However, on large instances, our
approach seems to be better, see for example ``facebook-like'', where we are
faster by a factor of roughly six, and ``infectious'', where the algorithm based
on static expansions did not finish within five hours.
We believe that the static-expansion-based algorithm has higher memory
requirements and more inefficient cache usage and hence becomes significantly
slower for larger instances.

\paragraph{Impact of temporal betweenness type on the temporal betweenness
distribution.} In \cref{fig:bethist1} we exemplarily show histograms of the
temporal betweenness values of the ``highschool-2013'' data set, the
``facebook-like'' data set, and the ``primaryschool'' data set. We can observe
that the temporal betweenness centrality has a power-law-like distribution and
the temporal betweenness type does not have a strong influence on the
distribution in the ``highschool-2013'' data set and the ``facebook-like'' data
set (which is the case in most data sets).  In the ``primaryschool'' data set,
however, the distributions for non-strict shortest and strict shortest temporal
betweenness are very similar but the other ones are quite different from each other.

    \pgfplotsset{
        compat=1.3,
    }

%    \pgfplotstableread{
%        0 35569 8842    134984
%        1 30428 4689    34077
%        2 32920 6207    73787
%        3 16462 7562    23496
%        4 12315 8572    16565
%        5 76572 19572   26030
%    }\dataset
    \begin{figure}[t!] \centering
\pgfplotstableread{
0 200 180 200 193 181
1 71 79 71 73 70
2 28 30 28 32 44
3 11 20 11 12 14
4 6 6 6 7 5
5 4 7 4 5 7
6 2 0 2 1 2
7 2 0 2 1 1
8 2 3 2 1 1
9 1 2 1 2 2
}\dataset

%    \pgfplotstableread{
%        0 35569 8842    134984
%        1 30428 4689    34077
%        2 32920 6207    73787
%        3 16462 7562    23496
%        4 12315 8572    16565
%        5 76572 19572   26030
%    }\dataset
\begin{tikzpicture}[scale=.8]
    \begin{axis}[
        width=18cm,
        height=5cm,
        ymin=0,
        ymax=200,
        %ymode=log,
        ybar,
        bar width=4pt,
        %ylabel={Y-Label},
        %xtick=data,
%        xticklabels={
%            1,
%            2,
%            3,
%            4,
%            5,
%            6,
%            7,
%            8,
%            9,
%            10
%        },
        %xticklabel style={yshift=-8ex},
        %major x tick style={
            % (this is a better way than assigning `opacity=0')
        %    /pgfplots/major tick length=0pt,
        %},
        %minor x tick num=1,
        %minor tick length=2ex,
        % ---------------------------------------------------------------------
        % (adapted solution from <https://tex.stackexchange.com/a/141006/95441>)
        % we want to provide absolute `at' values ...
%        scatter/position=absolute,
%        node near coords style={
%            % ... to provide axis coordinates at `ymin' for the nodes
%            at={(axis cs:\pgfkeysvalueof{/data point/x},\pgfkeysvalueof{/pgfplots/ymin})},
%            % then also the `anchor' has to be adapted ...
%            anchor=east,
%            % ... because we rotate the labels which would overlap otherwise
%            rotate=90,
%        },
        % ---------------------------------------------------------------------
        % (created a cycle list to shorten the below `\addplot' entries)
        cycle list={
            {draw=black,fill=red!80},
            {draw=black,fill=blue!80},
            {draw=black,fill=red!60},
            {draw=black,fill=blue!60},
            {draw=black,fill=blue!40},
        },
        % (moved common option here)
        table/x index=0,
    ]
        \addplot+ [] table [y index=1] \dataset;
        \addplot+ [] table [y index=2] \dataset;
        \addplot+ [] table [y index=3] \dataset;
        \addplot+ [] table [y index=4] \dataset;
        \addplot+ [] table [y index=5] \dataset;
    \end{axis}
\end{tikzpicture}    
    
    \pgfplotstableread{
0 1822 1799 1823 1799 1777
1 39 56 38 56 76
2 15 20 15 20 23
3 9 9 8 9 9
4 4 6 5 6 7
5 1 2 1 2 2
6 3 2 3 2 2
7 0 1 0 1 1
8 2 0 2 0 0
9 4 4 4 4 2
}\dataset
    
\begin{tikzpicture}[scale=.8]
    \begin{axis}[
        width=18cm,
        height=5cm,
        ymin=0,
        ymax=1900,
        ymode=log,
        ybar,
        bar width=4pt,
        %ylabel={Y-Label},
        %xtick=data,
%        xticklabels={
%            1,
%            2,
%            3,
%            4,
%            5,
%            6,
%            7,
%            8,
%            9,
%            10
%        },
        %xticklabel style={yshift=-8ex},
        %major x tick style={
            % (this is a better way than assigning `opacity=0')
        %    /pgfplots/major tick length=0pt,
        %},
        %minor x tick num=1,
        %minor tick length=2ex,
        % ---------------------------------------------------------------------
        % (adapted solution from <https://tex.stackexchange.com/a/141006/95441>)
        % we want to provide absolute `at' values ...
%        scatter/position=absolute,
%        node near coords style={
%            % ... to provide axis coordinates at `ymin' for the nodes
%            at={(axis cs:\pgfkeysvalueof{/data point/x},\pgfkeysvalueof{/pgfplots/ymin})},
%            % then also the `anchor' has to be adapted ...
%            anchor=east,
%            % ... because we rotate the labels which would overlap otherwise
%            rotate=90,
%        },
        % ---------------------------------------------------------------------
        % (created a cycle list to shorten the below `\addplot' entries)
        cycle list={
            {draw=black,fill=red!80},
            {draw=black,fill=blue!80},
            {draw=black,fill=red!60},
            {draw=black,fill=blue!60},
            {draw=black,fill=blue!40},
        },
        % (moved common option here)
        table/x index=0,
    ]
        \addplot+ [] table [y index=1] \dataset;
        \addplot+ [] table [y index=2] \dataset;
        \addplot+ [] table [y index=3] \dataset;
        \addplot+ [] table [y index=4] \dataset;
        \addplot+ [] table [y index=5] \dataset;
    \end{axis}
\end{tikzpicture}

\pgfplotstableread{
0 78 135 79 153 107
1 58 68 57 56 70
2 40 20 41 15 26
3 18 7 18 5 13
4 18 2 16 5 10
5 13 3 14 0 5
6 11 0 11 2 5
7 3 4 3 1 2
8 2 1 2 2 2
9 1 2 1 3 2
}\dataset

%    \pgfplotstableread{
%        0 35569 8842    134984
%        1 30428 4689    34077
%        2 32920 6207    73787
%        3 16462 7562    23496
%        4 12315 8572    16565
%        5 76572 19572   26030
%    }\dataset
\begin{tikzpicture}[scale=.8]
    \begin{axis}[
        width=18cm,
        height=5cm,
        ymin=0,
        ymax=160,
        %ymode=log,
        ybar,
        bar width=4pt,
        %ylabel={Y-Label},
        %xtick=data,
%        xticklabels={
%            1,
%            2,
%            3,
%            4,
%            5,
%            6,
%            7,
%            8,
%            9,
%            10
%        },
        %xticklabel style={yshift=-8ex},
        %major x tick style={
            % (this is a better way than assigning `opacity=0')
        %    /pgfplots/major tick length=0pt,
        %},
        %minor x tick num=1,
        %minor tick length=2ex,
        % ---------------------------------------------------------------------
        % (adapted solution from <https://tex.stackexchange.com/a/141006/95441>)
        % we want to provide absolute `at' values ...
%        scatter/position=absolute,
%        node near coords style={
%            % ... to provide axis coordinates at `ymin' for the nodes
%            at={(axis cs:\pgfkeysvalueof{/data point/x},\pgfkeysvalueof{/pgfplots/ymin})},
%            % then also the `anchor' has to be adapted ...
%            anchor=east,
%            % ... because we rotate the labels which would overlap otherwise
%            rotate=90,
%        },
        % ---------------------------------------------------------------------
        % (created a cycle list to shorten the below `\addplot' entries)
        cycle list={
            {draw=black,fill=red!80},
            {draw=black,fill=blue!80},
            {draw=black,fill=red!60},
            {draw=black,fill=blue!60},
            {draw=black,fill=blue!40},
        },
        % (moved common option here)
        table/x index=0,
    ]
        \addplot+ [] table [y index=1] \dataset;
        \addplot+ [] table [y index=2] \dataset;
        \addplot+ [] table [y index=3] \dataset;
        \addplot+ [] table [y index=4] \dataset;
        \addplot+ [] table [y index=5] \dataset;
    \end{axis}
\end{tikzpicture}

\caption{Top: Temporal betweenness histogram of ``highschool-2013'' data set.
Vertices are collected in 10 evenly distributed buckets between 0 and the
highest temporal betweenness value. The $y$-axis corresponds to the number of
vertices. Middle:
Temporal betweenness histogram of ``facebook-like'' data set. The $y$-axis is on a log-scale.
Bottom: Temporal betweenness histogram of ``primaryschool'' data set. 
% Here, the $y$-axis is \emph{not} on a log-scale.
The temporal betweenness types from left to right are: non-strict shortest,
non-strict shortest foremost, strict shortest, strict shortest foremost, strict
prefix foremost. Variants of foremost betweenness are
colored in shades of red; shortest betweenness is colored
in shades of blue.}\label{fig:bethist1}
\end{figure}
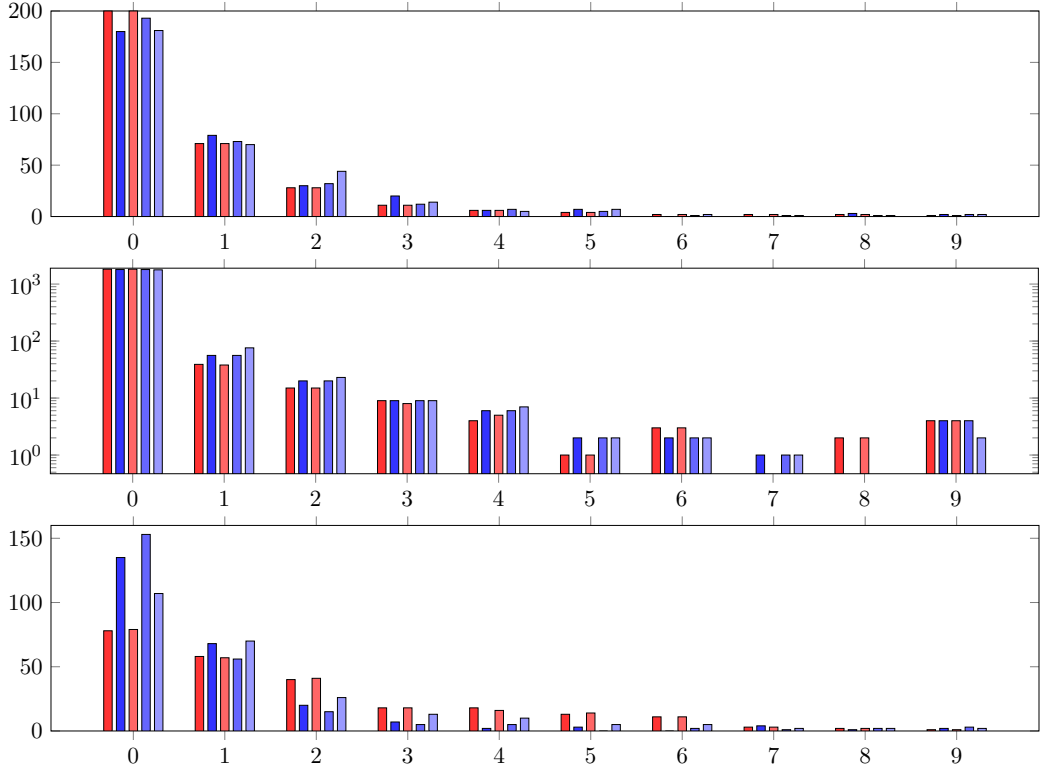

\paragraph{Impact of temporal betweenness type on the vertex ranking.}
\begin{table*}[t!]
\tiny
  \centering
  \caption{Kendall's tau correlation measure for each pair of vertex rankings induced by the different temporal betweenness versions. Values close to 1 indicate strong agreement and values close to -1 indicate strong disagreement between the rankings.}
  \makebox[\textwidth]{ 
  \begin{filecontents*}{results2.csv}
Data,top12,top13,top14,top15,top23,top24,top25,top34,top35,top45
highschool-2011,0.513904761904762,0.997206349206349,0.543111111111111,0.48192175050421,0.51415873015873,0.869714285714286,0.80549052075308,0.543365079365079,0.482175928642661,0.840058747582496
highschool-2012,0.488640595903166,0.99366852886406,0.48665425201738,0.431019509132708,0.48975791433892,0.93072625698324,0.810877958696798,0.487523277467411,0.432633814410359,0.819818726388401
highschool-2013,0.430442205587137,0.997636066865537,0.43089247856513,0.386416963061508,0.430329637342639,0.885255436108141,0.75434911680751,0.430779910320632,0.385929039160794,0.780058953114356
primaryschool,0.262919652961147,0.985940125510099,0.289981310376813,0.178628990775351,0.256815609889922,0.758834725167704,0.581975926751483,0.283191303345757,0.170604574603066,0.632430856902694
hospital-ward,0.636756756756757,0.997117117117117,0.630990990990991,0.546130908325284,0.638198198198198,0.968288288288288,0.816275876347149,0.632432432432432,0.547591151395672,0.822116848628703
hypertext,0.541719342604298,0.997787610619469,0.554993678887484,0.50742968075035,0.542035398230088,0.892541087231353,0.778691154945864,0.554677623261694,0.50774583631468,0.847613067969899
facebook-like,0.8354665863526,0.970997195124227,0.836156673546456,0.7492121907985,0.835379239614232,0.973540041940954,0.833652283338728,0.83663910518086,0.749304594523054,0.833338583120402

\end{filecontents*}
\pgfplotstabletypeset[
   col sep=comma,
   columns={Data,top12,top13,top14,top15,top23,top24,top25,top34,top35,top45}, %,Classical Degeneracy,Degeneracy0,Degeneracy3,Degeneracy5,Degeneracy7},
   columns/Data/.style={column type=l,string type, column name={Data Set}},
   columns/top12/.style={column type=r, column name={\begin{minipage}{12ex}non-str.\ sh vs.\\ non-str.\ sh fm\end{minipage}}},
   columns/top13/.style={column type=r,int detect, column name={\begin{minipage}{10ex}non-str.\ sh vs.\ str.\ sh\end{minipage}}},
   columns/top14/.style={column type=r,int detect, column name={\begin{minipage}{12ex}non-str.\ sh vs.\ str.\ sh fm\end{minipage}}},
   columns/top15/.style={column type=r,int detect, column name={\begin{minipage}{10ex}non-str.\ sh vs.\ str.\ p fm\end{minipage}}},
   columns/top23/.style={column type=r,int detect, column name={\begin{minipage}{12ex}non-str.\ sh fm vs.\ str.\ sh\end{minipage}}},
   columns/top24/.style={column type=r,int detect, column name={\begin{minipage}{12ex}non-str.\ sh fm vs.\ str.\ sh fm\end{minipage}}},
   columns/top25/.style={column type=r,int detect, column name={\begin{minipage}{12ex}non-str.\ sh fm vs.\ str.\ p fm\end{minipage}}},
   columns/top34/.style={column type=r,int detect, column name={\begin{minipage}{8ex}str.\ sh vs.\ str.\ sh fm\end{minipage}}},
   columns/top35/.style={column type=r,int detect, column name={\begin{minipage}{8ex}str.\ sh vs.\ str.\ p fm\end{minipage}}},
   columns/top45/.style={column type=r,int detect, column name={\begin{minipage}{10ex}str.\ sh fm vs.\ str.\ p fm\end{minipage}}},
   %columns/Classical Degeneracy/.style={column name={Static}},
   %columns/Degeneracy0/.style={column name={$\Delta=0$}},
   %columns/Degeneracy3/.style={column name={$\sim 5^3$}},
   %columns/Degeneracy5/.style={column name={$\sim 5^5$}},
   %columns/Degeneracy7/.style={column name={$\sim 5^7$}},
   every head row/.style={before row=\toprule}, %  & & & & & \multicolumn{5}{c}{Degeneracy} \\ \cmidrule(r){6-10}, after row=\midrule},
   every last row/.style={after row=\bottomrule},
   ]{results2.csv}}
  \label{tab:KT}
\end{table*}%
\begin{table*}[t!]
\scriptsize
  \centering
  \caption{Size of the intersection of each pair of sets containing the ten vertices with highest temporal betweenness values.}
  \makebox[\textwidth]{ 
  \begin{filecontents*}{results.csv}
Data,top12,top13,top14,top15,top23,top24,top25,top34,top35,top45
highschool-2011,4,10,5,4,4,8,8,5,4,9
highschool-2012,4,10,4,3,4,9,7,4,3,8
highschool-2013,4,10,4,3,4,8,5,4,3,7
primaryschool,0,9,0,0,0,9,7,0,0,8
hospital-ward,7,10,7,7,7,9,8,7,7,9
hypertext,4,10,4,4,4,10,9,4,4,9
facebook-like,9,10,9,6,9,10,7,9,6,7
\end{filecontents*}
\pgfplotstabletypeset[
   col sep=comma,
   columns={Data,top12,top13,top14,top15,top23,top24,top25,top34,top35,top45}, %,Classical Degeneracy,Degeneracy0,Degeneracy3,Degeneracy5,Degeneracy7},
   columns/Data/.style={column type=l,string type, column name={Data Set}},
columns/top12/.style={column type=r, column name={\begin{minipage}{12ex}non-str.\ sh vs.\\ non-str.\ sh fm\end{minipage}}},
   columns/top13/.style={column type=r,int detect, column name={\begin{minipage}{10ex}non-str.\ sh vs.\ str.\ sh\end{minipage}}},
   columns/top14/.style={column type=r,int detect, column name={\begin{minipage}{12ex}non-str.\ sh vs.\ str.\ sh fm\end{minipage}}},
   columns/top15/.style={column type=r,int detect, column name={\begin{minipage}{10ex}non-str.\ sh vs.\ str.\ p fm\end{minipage}}},
   columns/top23/.style={column type=r,int detect, column name={\begin{minipage}{12ex}non-str.\ sh fm vs.\ str.\ sh\end{minipage}}},
   columns/top24/.style={column type=r,int detect, column name={\begin{minipage}{12ex}non-str.\ sh fm vs.\ str.\ sh fm\end{minipage}}},
   columns/top25/.style={column type=r,int detect, column name={\begin{minipage}{12ex}non-str.\ sh fm vs.\ str.\ p fm\end{minipage}}},
   columns/top34/.style={column type=r,int detect, column name={\begin{minipage}{8ex}str.\ sh vs.\ str.\ sh fm\end{minipage}}},
   columns/top35/.style={column type=r,int detect, column name={\begin{minipage}{8ex}str.\ sh vs.\ str.\ p fm\end{minipage}}},
   columns/top45/.style={column type=r,int detect, column name={\begin{minipage}{10ex}str.\ sh fm vs.\ str.\ p fm\end{minipage}}},
   %columns/Classical Degeneracy/.style={column name={Static}},
   %columns/Degeneracy0/.style={column name={$\Delta=0$}},
   %columns/Degeneracy3/.style={column name={$\sim 5^3$}},
   %columns/Degeneracy5/.style={column name={$\sim 5^5$}},
   %columns/Degeneracy7/.style={column name={$\sim 5^7$}},
   every head row/.style={before row=\toprule}, %  & & & & & \multicolumn{5}{c}{Degeneracy} \\ \cmidrule(r){6-10}, after row=\midrule},
   every last row/.style={after row=\bottomrule},
   ]{results.csv}}
  \label{tab:top10}
\end{table*}%
In \cref{tab:KT} we present Kendall's tau correlation
measure\footnote{The \emph{Kendall's tau correlation
measure} is a measure to compare rankings. It is
defined as follows: $(\text{\#\,concordant pairs} - \text{\#\,discordant
pairs}) \ / \ \text{\#\,pairs}$.}~\cite{knight1966computer} for each pair of
vertex rankings induced by the different temporal betweenness versions. For the
temporal betweennes variants ``shortest'' and ``shortest foremost'' we can
observe that their respective non-strict and strict variants produce very
similar rankings. We can further see that the betweenness variants ``non-strict
shortest foremost'', ``strict shortest foremost'', and ``strict prefix
foremost'' are pairwise similar.
Recall that strict prefix foremost temporal betweenness was computed
significantly faster by our algorithms, hence we can conclude that if foremost
temporal paths are of interest, then the prefix foremost temporal betweenness
is much easier to compute and yields very similar results.

These findings are supported by the pairwise comparisons of the sets of the ten
vertices with the largest temporal betweenness values, which are presented in
\cref{tab:top10}. Notably, the ``primaryschool'' has quite different top ten
vertex sets for all pairs of betweenness variants that do not fall into the
above mentioned pairings.
We can also observe that ``non-strict shortest'' and ``strict prefix foremost''
produce the most dissimilar vertex rankings and top ten sets.

\section{Conclusion}
	\label{ch:conclusion}
We investigated several variants of temporal betweenness centrality based on the various optimization criteria for temporal paths. We have shown a surprising discrepancy in their computational complexity: while some variations are \SPC, others can be computed in polynomial time. More specifically, we found that 
counting foremost, and thus fastest paths, is \SPC, and in turn, the computation of the corresponding betweenness centrality scores is \SPC\ as well. 
In contrast to that, one can count shortest and shortest foremost temporal paths in polynomial time both for strict and non-strict paths. 
In the case of prefix-foremost paths, however, we found a polynomial-time algorithm for the strict version, whereas the non-strict version is again \SPC.
An intuitive explanation for this behavior might be that our algorithms strongly
rely on a recursive formulation for the so-called temporal dependencies, which in turn
requires that the predecessor relation of optimal temporal paths is acyclic and
that prefixes of optimal temporal paths are also optimal. For all
optimal temporal path types for which we show computational hardness of the
corresponding counting problem, one of the two mentioned requirements is not given.

As to challenges for future research, one direction is to attack the temporal
betweenness centrality variants which we have shown to be computationally hard
by means of approximation algorithms or the development of fixed-parameter
algorithms (e.g.\ for some structural graph parameters of the underlying graph).
Roughly in the same direction would be to undertake a closer study of the
special structures of real-world networks that might be algorithmically
exploitable.
A line of research directly motivated by our experiments would be to 
explore how to further decrease the memory consumption of our algorithms---the 
practical limitations seem to mainly come from the high space consumption,
at some point leading our algorithms to do a lot of paging (for the
static-expansion based algorithm, this point seems to be reached much earlier).
Specifically, our algorithm for shortest and shortest foremost temporal betweenness stores
dependency and path count values for every vertex-time step combination. This is
major difference to the static algorithm of Brandes~\cite{brandes_faster_2001}
which is less vulnerable in this respect.

On the experimental side, we compared our the running times of our
approach to static expansion based temporal betweenness algorithms. 
We found that for small instances, the static expansion based temporal
betweenness algorithms are faster but on large instances or algorithm performs
better. We conjecture that this is the case because our algorithm has a better
memory usage, hence the specific threshold of the input size where our
algorithm is faster probably depends heavily on the machine that is used for
the computation.

We can further observe that the temporal betweenness type does
not have a strong impact on the distribution of the betweenness values. When
comparing the vertex rankings produced by the betweenness values we can observe
that there is little difference between strict and the respective non-strict
variants. Furthermore, ``shortest foremost'' and ``prefix foremost'' yield
similar results in terms of the betweenness values of the vertices, while the
strict prefix foremost temporal betweenness values can be computed significantly faster.

It would
be very enlightening to compare the betweenness values of the two tractable
variants of temporal betweenness based on foremost paths (shortest foremost and
prefix foremost) to the betweenness values for temporal betweenness based on
foremost temporal paths (which is intractable). However, that requires a way to
compute or approximate the temporal betweenness values based on
foremost temporal paths. Hence, developing efficient (parameterized 
exponential-time)
algorithms or approximation algorithms for the temporal betweenness based
on foremost temporal paths is also an interesting future work direction.

\bibliographystyle{abbrvnat}
\bibliography{literature}

\end{document}